\documentclass[runningheads]{llncs}
\usepackage{graphicx}
\usepackage{amssymb}
\usepackage{amsmath}
\usepackage{comment}
\usepackage{todonotes}
\usepackage[utf8]{inputenc}
\usepackage{shuffle}
\usepackage{multirow}
\usepackage{listings}
\usepackage{mathabx}

\usepackage{thmtools, thm-restate}
\usepackage{hyperref}
\usepackage[sectionbib, square,sort,comma,numbers]{natbib}

\usetikzlibrary{shapes,arrows,automata,positioning}

\newcommand{\perm}{\operatorname{perm}}

\DeclareSymbolFont{rsfscript}{OMS}{rsfs}{m}{n}
\DeclareSymbolFontAlphabet{\mathrsfs}{rsfscript}

\newcommand{\NP}{\textsf{NP}}
\newcommand{\PSPACE}{\textsf{PSPACE}}
\newcommand{\NPSPACE}{\textsf{NPSPACE}}
\newcommand{\PTIME}{\textsf{P}}

\usepackage{tabularx,lipsum,environ,amsmath,amssymb}

\newcounter{problemcounter}
\makeatletter
\newcommand{\problemtitle}[1]{\gdef\@problemtitle{#1}}
\newcommand{\probleminput}[1]{\gdef\@probleminput{#1}}
\newcommand{\problemquestion}[1]{\gdef\@problemquestion{#1}}
\NewEnviron{decproblem}{
  \refstepcounter{problemcounter}
  \problemtitle{}\probleminput{}\problemquestion{}
  \BODY
  \par\addvspace{.5\baselineskip}
  \noindent
  \begin{tabularx}{\textwidth}{@{\hspace{\parindent}} l X c}
    \multicolumn{2}{@{\hspace{\parindent}}l}{\textbf{Decision Problem \theproblemcounter:} \@problemtitle} \\
    \textbf{Input:} & \@probleminput \\
    \textbf{Question:} & \@problemquestion
  \end{tabularx}
  \par\addvspace{.5\baselineskip}
}
\makeatother

\begin{document}
\title{Computational Complexity of Synchronization under Regular Commutative Constraints}
\titlerunning{Synchronization under Regular Commutative Constraints}
%
%
\author{Stefan Hoffmann\orcidID{0000-0002-7866-075X}}
\authorrunning{S. Hoffmann}
%
\institute{Informatikwissenschaften, FB IV, 
  Universit\"at Trier,  Universitätsring 15, 54296~Trier, Germany, 
  \email{hoffmanns@informatik.uni-trier.de}}
\maketitle              
%

 
\begin{abstract}

Here we study the computational complexity of
the constrained synchronization problem
for the class of regular commutative constraint languages.
Utilizing
a vector representation of regular commutative constraint languages, 
we give a full classification of the computational complexity 
of the constrained synchronization problem.
Depending on the constraint language, our problem becomes $\PSPACE$-complete, $\NP$-complete 
or polynomial time solvable. In addition, we derive a polynomial
time decision procedure for the complexity of the constrained
synchronization problem, given a constraint automaton 
accepting a commutative language as input.

\end{abstract}

\section{Introduction}

A deterministic semi-automaton is synchronizing if it admits a reset word, i.e., a word which leads to a definite
state, regardless of the starting state. This notion has a wide range of applications, from software testing, circuit synthesis, communication engineering and the like, see \cite{Vol2008,San2005}.  The famous \v{C}ern\'y conjecture \cite{Cer64}
states that a minimal synchronizing word has at most quadratic length. 
We refer to the mentioned survey articles for details. 
Due to its importance, the notion of synchronization has undergone a range of generalizations and variations
for other automata models.
It was noted in \cite{Martyugin12} that in some  generalizations only certain paths, or input words, are allowed (namely those for which the input automaton is defined). In \cite{Gusev:2012}
the notion of constrained synchronization was 
introduced in connection with a reduction procedure
for synchronizing automata.
The paper \cite{DBLP:conf/mfcs/FernauGHHVW19} introduced the computational problem of constrained 
synchronization. In this problem, we search for a synchronizing word coming from a specific subset of allowed
input sequences. For further motivation and applications we refer to the aforementioned paper \cite{DBLP:conf/mfcs/FernauGHHVW19}.
In this paper, a complete analysis of the complexity landscape when the constraint language is given by small partial automata was done. It is natural to extend this result to other language classes, or
even to give a complete classification of all the complexity classes that could arise.
Our work is in this vein, we will look at the complexity landscape for commutative regular constraint languages.


\section{Prerequisites}
\label{sec::prerequisites}

\subsection{General Notions and Problems Related to Automata and Synchronization}

By $\mathbb N_0 = \{0,1,2,\ldots\}$ we denote the natural numbers with zero. 
Setting $n < \infty$ for all $n \in \mathbb N_0$,
we will use the symbol $\infty$ in connection with $\mathbb N_0$.
Hence we regard $\mathbb N_0 \cup \{\infty\}$
as an ordered set with top element $\infty$.
Throughout the paper, we consider deterministic finite automata (DFAs).
Recall that a DFA~$\mathcal A$ is a tuple $\mathcal A = (\Sigma, Q, \delta, q_0, F)$,
where the alphabet $\Sigma$ is a finite set of input symbols,~$Q$ is the finite state set, with start state $q_0 \in Q$, and final state set $F \subseteq Q$.
The transition function $\delta : Q\times \Sigma \to Q$ extends to words from $\Sigma^*$ in the usual way. The function $\delta$ can be further extended to sets of states in the following way. For every set $S \subseteq Q$ with $S \neq \emptyset$ and $w \in \Sigma^*$, we set $\delta(S, w) := \{\,\delta(q, w) \mid q \in S\,\}$.
We call $\mathcal A$ \emph{complete} if~$\delta$ is defined for every $(q,a)\in Q \times \Sigma$; if $\delta$ is undefined for some $(q,a)$, the automaton~$\mathcal A$ is called \emph{partial}.
If $|\Sigma| = 1$, we call $\mathcal A$ a \emph{unary} automaton. 
The set $L(\mathcal A) = \{\, w \in \Sigma^* \mid \delta(q_0, w) \in F\,\}$ denotes the language
accepted by $\mathcal A$.
A semi-automaton is a finite automaton without a specified start state
and with no specified set of final states.
The properties of being \emph{deterministic}, \emph{partial}, and \emph{complete} for semi-automata are defined as for DFAs.
When the context is clear, we call both deterministic finite automata and semi-automata simply \emph{automata}.
An automaton $\mathcal A$ is called \emph{synchronizing} if there exists a word $w \in \Sigma^*$ with $|\delta(Q, w)| = 1$. In this case, we call $w$ a \emph{synchronizing word} for $\mathcal A$.
We call a state $q\in Q$ with $\delta(Q, w)=\{q\}$ for some $w\in \Sigma^*$ a \emph{synchronizing state}.

\begin{theorem}\cite{Vol2008} \label{thm:unrestricted_sync_poly_time}
	For any deterministic complete semi-automaton, we
	can decide if it is synchronizing in polynomial time $O(|\Sigma||Q|^2)$.
	Additionally, if we want to compute a synchronizing word $w$, then we can do this in time~$O(|Q|^3 + |Q|^2|\Sigma|))$ 
	and the length of $w$ will be $O(|Q|^3)$.
\end{theorem}
The following obvious remark, stating that the set of synchronizing words
is a two-sided ideal, will be used frequently without further mentioning.

\begin{lemma}
	\label{lem:append_sync} 
	Let $\mathcal A = (\Sigma, Q, \delta)$ be a deterministic and complete semi-automaton and $w\in \Sigma^*$ be a synchronizing word for $\mathcal A$. Then for every $u, v \in \Sigma^*$, the word $uwv$ is also synchronizing for~$\mathcal A$. 
\end{lemma}

We assume the reader to have some basic knowledge in computational complexity theory and formal language theory, as contained, e.g., in~\cite{HopMotUll2001}. For instance, we make use of  regular expressions to describe languages.
For a word $w \in \Sigma^*$ we denote by $|w|$ its length,
and for a symbol $x \in \Sigma$ we write $|w|_x$ to denote the number of occurences of $x$
in the word. We denote the empty word, i.e., the word of length zero, by $\varepsilon$.
We also make use of complexity classes like $\PTIME$, $\NP$, or $\PSPACE$.
With  $\leq^{\log}_m$ we denote a logspace many-one reduction.
If for two problems $L_1, L_2$ it holds that $L_1 \leq^{\log}_m L_2$ and $L_2 \leq^{\log}_m L_1$, then we write $L_1 \equiv^{\log}_m L_2$.
In~\cite{DBLP:conf/mfcs/FernauGHHVW19} the \emph{constrained synchronization problem}
was defined for a fixed partial deterministic automaton
$\mathcal B = (\Sigma, P, \mu, p_0, F)$. 

\begin{decproblem}\label{def:problem_L-constr_Sync}
  \problemtitle{\cite{DBLP:conf/mfcs/FernauGHHVW19}~\textsc{$L(\mathcal B)$-Constr-Sync}}
  \probleminput{Deterministic complete semi-automaton $\mathcal A = (\Sigma, Q, \delta)$.}
  \problemquestion{Is there a synchronizing word $w \in \Sigma^*$ for $\mathcal A$ with  $w \in L(\mathcal B)$?}
\end{decproblem}

The automaton $\mathcal B$ will be called the \emph{constraint automaton}.
If an automaton $\mathcal A$ is a yes-instance of \textsc{$L(\mathcal B)$-Constr-Sync} we call $\mathcal A$ \emph{synchronizing with respect to $\mathcal{B}$}. 
Occasionally,
we do not specify $\mathcal{B}$ and rather talk about \textsc{$L$-Constr-Sync}.

A language $L \subseteq \Sigma^*$ is called \emph{commutative}
if with $w \in L$, every word 
arising out of $w$ by permuting its letters is also in $L$.
Essentially, a commutative language is defined by conditions that say how often a letter
is allowed to appear in its words, but not by the actual position of that letter.
For this class of languages it was noted that it is structurally simple \cite{DBLP:conf/cai/Hoffmann19,Hoffmann20}. Also  
in terms of synchronizing words 
this class yields quite simple automata \cite{FernauHoffmann19}, but nevertheless may give algorithmic hard problems, as
this class is sufficient for many reductions \cite{FernauHoffmann19}. Here, we are concerned with
$L\textsc{-Constr-Sync}$ for the case that
the constraint language $L$ is a commutative regular language.
We will use the shuffle operation in connection with unary languages frequently to write commutative languages.

\begin{definition} The \emph{shuffle operation}, denoted by $\shuffle$, is defined as
 \begin{align*}
    u \shuffle v & := \left\{ \begin{array}{ll}
     \multirow{2}{*}{$x_1 y_1 x_2 y_2 \cdots x_n y_n  \mid$} &  u = x_1 x_2 \cdots x_n, v = y_1 y_2 \cdots y_n, \\ 
         &   x_i, y_i \in \Sigma^{\ast}, 1 \le i \le n, n \ge 1
  \end{array} \right\},
 \end{align*}
 for $u,v \in \Sigma^{\ast}$ and 
  $L_1 \shuffle L_2  := \bigcup_{x \in L_1, y \in L_2} (x \shuffle y)$ for $L_1, L_2 \subseteq \Sigma^{\ast}$.
\end{definition}


\subsection{Unary Languages}

Let $\Sigma = \{a\}$ be a unary alphabet. Suppose $L \subseteq \Sigma^{\ast}$ is regular
with an accepting complete deterministic automaton $\mathcal A = (\Sigma, S, \delta, s_0, F)$. Then by considering
the sequence of states $\delta(s_0, a^1), \delta(s_0, a^2), \delta(s_0, a^3), \ldots$ we find numbers $i \ge 0, p > 0$ with $i+p$ minimal such that $\delta(s_0, a^i) = \delta(s_0, a^{i+p})$.
We call these numbers the index $i$ and the period $p$ of the automaton $\mathcal A$.
If $Q = \{\delta(s_0, a^m) \mid m \ge 0 \}$, then  $i + p = |S|$.
 In our discussion unary languages that are accepted by
 automata with a single final state appear.
 
 \begin{restatable}{lemma}{unarysinglefinal} \cite{DBLP:conf/cai/Hoffmann19}
\label{lem::unary_single_final}
  Let $L \subseteq \{a\}^{\ast}$ be a unary language that is accepted
  by an automaton with a single final state, index $i$ and period $p$.
  Then either $L = \{u\}$ with $|u| < i$ (and if the automaton is minimal we would have $p = 1$),
  or $L$ is infinite with $L = a^{i+m}(a^p)^{\ast}$ and $0 \le m < p$. Hence
  two words $u,v$ with $\min\{|u|, |v|\} \ge i$ are both in $L$ or not if and only
  if $|u| \equiv |v| \pmod{p}$.
 \end{restatable}

\subsection{Known Result on Constrained Synchronization
 and Commutative Languages} 

Here we collect results from \cite{DBLP:conf/mfcs/FernauGHHVW19,DBLP:conf/cai/Hoffmann19},
and some consequences that will be used later.
First a mild extension of a lemma from \cite{DBLP:conf/mfcs/FernauGHHVW19}, where it was
formulated only for the class $\PTIME$, but it also holds for $\NP$
and $\PSPACE$.

\begin{restatable}[]{lemma}{lemunion} 
	\label{lem:union}
	Let $\mathcal X$ denote any of the complexity classes
	$\PTIME$, $\NP$ or $\PSPACE$.
	If $L(\mathcal B)$ is a finite union of languages $L(\mathcal B_1),
	L(\mathcal B_2), \dots, L(\mathcal B_n)$ such that for each $1\leq i\leq n$
the problem $L(\mathcal B_i)\textsc{-Constr-Sync}\in \mathcal X$, 
	then $L(\mathcal B)\textsc{-Constr-Sync}\in~\mathcal X$.
\end{restatable}

The next result from \cite{DBLP:conf/mfcs/FernauGHHVW19}
states that the computational complexity is always in $\PSPACE$.

\begin{theorem} \cite{DBLP:conf/mfcs/FernauGHHVW19}
  \label{thm:L-contr-sync-PSPACE}
  For any constraint automaton $\mathcal B = (\Sigma, P, \mu, p_0, F)$
  the problem \textsc{$L(\mathcal B)$-Constr-Sync} is in $\PSPACE$.
\end{theorem}

If $|L(\mathcal B)| = 1$, then $L(\mathcal B)\textsc{-Constr-Sync}$
is obviously in $\PTIME$. Simply feed this single word into the input
semi-automaton for every state and check if a unique state results.
Hence by Lemma \ref{lem:union} the next is implied.

\begin{restatable}[]{lemma}{lemfinite}\label{lem:finite} 
 Let $\mathcal B = (\Sigma, P, \mu, p_0, F)$ be a constraint automaton
 such that $L(\mathcal B)$ is finite, then
 $L(\mathcal B)\textsc{-Constr-Sync} \in \PTIME$.
\end{restatable}

The following result from \cite{DBLP:conf/mfcs/FernauGHHVW19}
gives a criterion for containment in $\NP$. 

\begin{restatable}[]{theorem}{thmgeninNP}\cite{DBLP:conf/mfcs/FernauGHHVW19}
\label{thm:gen:inNP}
    Let $\mathcal{B} = (\Sigma, P, \mu, p_0, F)$ be a partial deterministic finite automaton.
	Then,  $L(\mathcal B)\textsc{-Constr-Sync}\in\NP$ if  there is a $\sigma\in \Sigma$ such that for all states $p\in P$, if $L(\mathcal{B}_{p,\{p\}})$ is infinite, then  $L(\mathcal{B}_{p,\{p\}})\subseteq \{\sigma\}^*$.
\end{restatable}

With this we can deduce another sufficient condition for containment in
$\NP$, which is more suited for commutative languages.

\begin{restatable}[]{lemma}{lemgeninNP}\label{lem:gen:inNP}
 Let $\Sigma$ be our alphabet and suppose $a \in \Sigma$.
 If
 $$
  L = \{a\}^* \shuffle F_1 \shuffle \ldots \shuffle F_k
 $$
 for finite languages $F_1, \ldots, F_k$, then
 $L\textsc{-Constr-Sync} \in \NP$.
\end{restatable}

The next result from~\cite{DBLP:conf/mfcs/FernauGHHVW19}
will be useful in making several simplifying assumptions about the constraint language
later in Section~\ref{sec:simplification}.

\begin{theorem}\cite{DBLP:conf/mfcs/FernauGHHVW19}
	\label{thm:add-stuff}
	Let $L\subseteq L'\subseteq\Sigma^*$. If  $L'\subseteq\{\,v\in \Sigma^*\mid \exists u,w \in \Sigma^*: uvw\in L\,\}$, then $L\textsc{-Constr-Sync}\equiv^{\log}_m L'\textsc{-Constr-Sync}$. 
\end{theorem}

The following Theorem~\ref{thm:reg_commutative_form}
is taken from~\cite{DBLP:conf/cai/Hoffmann19} and will be crucial 
in deriving our vector representation form for the constraint language later in Section~\ref{sec:simplification}.

\begin{theorem}\label{thm:reg_commutative_form}
 Let $\Sigma = \{a_1, \ldots, a_k\}$ be our alphabet.
 A commutative language $L \subseteq \Sigma^*$
 is regular if and only if it could be written in the form
 $$
  L = \bigcup_{i=1}^n U_1^{(i)} \shuffle \ldots \shuffle U_k^{(i)}
 $$
 with non-empty unary regular languages $U_j^{(i)} \subseteq \{a_j\}^*$
 for $i \in \{1,\ldots, n\}$ and $j \in \{1,\ldots k\}$
 that could be accepted
 by a unary automaton with a single final~state.
\end{theorem}

With respect to the Constrained Synchronization Problem~\ref{def:problem_L-constr_Sync}, 
for commutative constraint languages $L(\mathcal B)$, 
we will refer more to the form given by Theorem~\ref{thm:reg_commutative_form}
than to the specific automaton $\mathcal B = (\Sigma, P, \mu, p_0, F)$ underlying it.
In Section~\ref{sec:decision_problem} we will give some details how to compute such a form
for a given automaton accepting a commutative language.

%

\section{Results}

Our main result, Theorem~\ref{thm:complete_classification}, gives a complete classification
of the computational complexity of $L\textsc{-Constr-Sync}$, for different
regular commutative constraint languages. In the following sections, we will
prove various simplifications, propositions, corollaries and lemmata that ultimately
will all be used in proving Theorem \ref{thm:complete_classification}.
First, we will give criteria that allow certain simplification
of the constraint language, and derive a mechanism to describe
a given constraint language by a set of vectors, which gives all the essential information 
with regard to our problem. This notion will be used repeatedly
in all the following arguments. In Section \ref{sec:poly} 
we will give sufficient conditions
for containment in $\PTIME$. Then
we single out those instances that give hardness results for the complexity classes
$\NP$ and $\PSPACE$ in Section~\ref{sec:NP} and Section~\ref{sec:PSPACE}. Finally, in Section~\ref{sec:main_thm},
we combine all these results to prove Theorem~\ref{thm:complete_classification}.
From Theorem~\ref{thm:complete_classification}, in the last Section~\ref{sec:decision_problem}, 
a decision procedure
is derived to decide the complexity of $L(\mathcal B)\textsc{-Constr-Sync}$,
if we allow $\mathcal B$ to be part of our input.

\subsection{Simplifications of the Constraint Language}
\label{sec:simplification}

 Our first Proposition \ref{prop:infinite_then_N}
 follows from Theorem \ref{thm:add-stuff}. Very roughy, it says
 that for the letters that are allowed infinitely often, the exact way in which
 they appear is not that important, but only that we can find arbitrary long
 sequences of them.
 We then use
 this result to derive a more compact description, in terms
 of vectors over $\mathbb N_0 \cup \{\infty\}$,
 to capture the essential part of a commutative constraint language $L$
 with respect to the problem $L\textsc{-Constr-Sync}$.

\begin{restatable}[]{proposition}{infinitethenN}{(infinite language simplification)} \label{prop:infinite_then_N}
 Let $\Sigma = \{a_1, \ldots, a_k\}$ be our alphabet.
 Consider the Constrained Synchronization Problem \ref{def:problem_L-constr_Sync}
 with commutative constraint language $L$.
 Suppose
 $$
  L = \bigcup_{i=1}^n U_1^{(i)} \shuffle \ldots \shuffle U_k^{(i)}
 $$
 with unary languages $U_j^{(i)} \subseteq \{a_j\}^*$
 for $i \in \{1,\ldots, n\}$ and $j \in \{1,\ldots k\}$.
 If for some $i_0 \in \{1,\ldots, n\}$ and $j_0 \in \{1,\ldots k\}$
 the unary language $U_{j_0}^{(i_0)}$ is infinite,
 then construct the new language
 $$
  L' =  \bigcup_{i=1}^n V_1^{(i)} \shuffle \ldots \shuffle V_k^{(i)}
 $$
 with 
 $$
  V_j^{(i)} = \left\{ \begin{array}{ll}
     \{a_j\}^* & \mbox{ if } i = i_0 \mbox{ and } j = j_0 \\
     U_j^{(i)} & \mbox{ otherwise}.
   \end{array}
  \right.
 $$
 We simply change the single language $U_{j_0}^{(i_0)}$ 
 for the language $\{a_j\}^*$.
 Then a complete and deterministic  
 input semi-automaton $\mathcal A = (\Sigma, Q, \delta)$
 has a synchronizing word in $L$
 if and only if it has one in $L'$ and $L\textsc{-Constr-Sync} \equiv_m^{\log} L'\textsc{-Constr-Sync}$.
\end{restatable}

Suppose $L$ is a constraint language
with 
$$
  L = \bigcup_{i=1}^n U_1^{(i)} \shuffle \ldots \shuffle U_k^{(i)}
$$
according to Theorem \ref{thm:reg_commutative_form}.
By Proposition \ref{prop:infinite_then_N},
for our purposes we can assume that if
 $U_j^{(i)}$ is infinite, then it has the
form $U_j^{(i)} = \{a_j\}^*$.
The unary languages $U_j^{(i)}$ for $j \in \{1,\ldots, k\}$
and $i \in \{1,\ldots, n\}$ 
are acccepted
by some unary automaton with a single final state. 
By Lemma \ref{lem::unary_single_final}, if such a language is non-empty
and finite it contains only a single word.
Hence, the only relevant information is whether such a unary language
part is infinite or what length has the single
unary word it contains. This is captured by the next definition.

\begin{definition}{(vector representation of $L$)}
\label{def:constraint_vector} 
%
%
 Let $\Sigma = \{a_1, \ldots, a_k\}$ be our alphabet.
 Consider the Constrained Synchronization Problem \ref{def:problem_L-constr_Sync}
 with commutative regular constraint language $L$.
 Suppose 
 \begin{equation}\label{eqn:vec_rep_def_form}
  L = \bigcup_{i=1}^n U_1^{(i)} \shuffle \ldots \shuffle U_k^{(i)}
 \end{equation}
 with non-empty unary languages $U_j^{(i)} \subseteq \{a_j\}^*$
 for $i \in \{1,\ldots, n\}$ and $j \in \{1,\ldots k\}$
 that are acceptable by unary automata with a single final 
 state.  Then we say that a set of vectors $N \subseteq (\mathbb N_0 \cup \{\infty\})^k$ 
 corresponds to $L$, according to Equation \eqref{eqn:vec_rep_def_form}, if
 $
  N = \{ (n^{(i)}_1, \ldots, n^{(i)}_k) \mid  i \in \{1,\ldots, n\} \} 
 $
 with\footnote{Note that, as by assumption,
 the languages $U_j^{(i)}$ for $i \in \{1,\ldots, n\}$ and $j \in\{1,\ldots, k\}$
 are accepted by unary automata with a single final state, by Lemma \ref{lem::unary_single_final},
 they only contain a single
 word if they are finite and non-empty.}
 $$
  n_j^{(i)} = \left\{
  \begin{array}{ll}
   \infty & \mbox{ if } U_j^{(i)} \mbox{ is infinite }, \\
   |u|    & \mbox{ if } U_j^{(i)} = \{ u \}
  \end{array}
  \right.
 $$
 for $i \in \{1,\ldots, n\}$ and $j \in \{1,\ldots, k\}$.
 By Theorem \ref{thm:reg_commutative_form}, every regular commutative 
 constraint language has at least one vector representation.
 \end{definition} 

{\footnotesize  
\begin{example} \label{ex:vector}Let $\Sigma = \{a,b,c\}$ with $a = a_1, b = a_2, c = a_3$. For the language $L = \{aa\} \shuffle b^* \cup \{a\} \shuffle \{bb\}\shuffle c(cc)^*$ we have $N = \{ (2, \infty, 0), (1, 2, \infty)\}$.
Please see Example~\ref{ex:classification} for other languages.
\end{example}
}

The language $L$ is infinite precisely if for some vector
at least one entry equals~$\infty$.
Another important observation, quite similar to Proposition \ref{prop:infinite_then_N},
allows us to make further assumptions about the constraint language, or
the vectors corresponding to it. It will be used in the proofs
of Proposition~\ref{prop:single_letter_unbounded_two_bounded}
and Proposition~\ref{prop:one_letter_bounded_two_unbounded}. 

\begin{restatable}[]{proposition}{vectorrepincomparable}{(comparable vectors simplification)}\label{prop:vector_rep_incomparable}
 Let $\Sigma = \{a_1, \ldots, a_k\}$.
 Consider $L\textsc{-Constr-Sync}$.
 Suppose $L$ has the form stated in Theorem \ref{thm:reg_commutative_form},
 \begin{equation}\label{eqn:vec_rep_form}
  L = \bigcup_{i=1}^n U_1^{(i)} \shuffle \ldots \shuffle U_k^{(i)}
 \end{equation}
 with unary languages $U_j^{(i)} \subseteq \{a_j\}^*$
 for $i \in \{1,\ldots, n\}$ and $j \in \{1,\ldots k\}$.
 Let $N$ be the vector set, corresponding to Equation \eqref{eqn:vec_rep_form} and according to Definition \ref{def:constraint_vector}.
 Suppose $x,y \in N$ with $x \le y$ and $x = (x_1^{(i_0)}, \ldots, x_k^{(i_0)})$
 for $i_0 \in \{1, \ldots, n\}$, i.e., the vector $x$
 arises out of the part $U_1^{(i_0)} \shuffle \ldots \shuffle U_k^{(i_0)}$
 in the above union for $L$.
 Construct the new language
 $$
  L' = \bigcup_{i \in \{1,\ldots n\}\setminus\{i_0\}} U_1^{(i)}\shuffle \ldots \shuffle U_k^{(i)}
 $$
 without the part $U_1^{(i_0)} \shuffle \ldots \shuffle U_k^{(i_0)}$.
 Then a complete and deterministic  
 input semi-automaton $\mathcal A = (\Sigma, Q, \delta)$
 has a synchronizing word in $L$
 if and only if it has one in $L'$
 and $L\textsc{-Constr-Sync} \equiv_m^{\log} L'\textsc{-Constr-Sync}$.
\end{restatable}

\begin{example}
 Let $\Sigma = \{a,b,c\}$ with $a = a_1$, $b = a_2$, $c = a_3$.
 If $L = aaa^* \shuffle \{b\} \cup a^* \shuffle \{bb\} \shuffle \{c\} \cup \{a\}$,
 then $N = \{ (\infty, 1, 0), (\infty, 2, 1), (1,0,0) \}$.
 After simplification by Proposition \ref{prop:vector_rep_incomparable}
 and Proposition \ref{prop:infinite_then_N},
 we get a computationally equivalent constrained synchronization problem,
 with constraint language $L' = a^* \shuffle \{bb\}\shuffle \{c\}$
 and vector representation $N' = \{(\infty, 2, 1)\}$. In this case $N'$ contains precisely the maximal vector in $N$.
\end{example}

Hence, by taking the maximal vectors, which does not change the complexity, 
we can assume that the vectors associated with any regular commutative
constraint language are pairwise incomparable.

\subsection{The Polynomial Time Solvable Variants of the Problem}
\label{sec:poly}

If in the sets $U_1^{(i)} \shuffle \ldots \shuffle U_k^{(i)}$
each $U_j^{(i)}$ is either infinite or $U_j^{(i)} = \{\varepsilon\}$,
then $L\textsc{-Constr-Sync}\in \PTIME$.

\begin{restatable}[]{proposition}{vectorzeroinfty}\label{prop:vector_zero_infty}
 Let $\Sigma = \{a_1, \ldots, a_k\}$ be our alphabet.
 Consider the Constrained Synchronization Problem \ref{def:problem_L-constr_Sync}.
 Suppose the commutative constraint language $L$ is
 decomposed as stated in Theorem \ref{thm:reg_commutative_form}, 
 \begin{equation}\label{eqn:union_L_infty_or_zero_case}
  L = \bigcup_{i=1}^n U_1^{(i)} \shuffle \ldots \shuffle U_k^{(i)}.
 \end{equation}
 Denote by $N = \{ (n_1^{(i)}, \ldots, n_k^{(i)} \mid i = 1, \ldots, n \}$
 the vector representation, according 
 to Definition \ref{def:constraint_vector} and corresponding 
 to Equation \eqref{eqn:union_L_infty_or_zero_case}.
 If for all $i \in \{1,\ldots, n\}$ and all $j \in \{1,\ldots, k\}$ 
 we have $n_j^{(i)} \in \{0,\infty\}$,
 then the problem is in \PTIME. 
\end{restatable}


Interestingly, because of Lemma~\ref{lem:length_sync_from_word_unary} stated next, if 
in the sets $U_1^{(i)} \shuffle \ldots \shuffle U_k^{(i)}$,
we have at most one $j_0 \in \{1,\ldots, k\}$ such that $U_{j_0}^{(i)} = \{ a_{j_0} \}$,
and at most one other $j_1 \in \{1,\ldots, k\}$ such that $U_{j_1}^{(i)}$
is infinite, and $U_{j}^{(i)} = \{\varepsilon\}$ for all $j \in \{1,\ldots, k\} \setminus \{j_0, j_1\}$,
then also $L\textsc{-Constr-Sync}\in \PTIME$.
Later, we will see that only a slight relaxation of this condition,
for example, if instead $U_{j_0}^{(i)} = \{ a_{j_0}a_{j_0} \}$
in the above, then the problem becomes \NP-complete.

\begin{restatable}[]{lemma}{lengthsyncfromwordunary} 
\label{lem:length_sync_from_word_unary}
 Let $\mathcal A = (\Sigma, Q, \delta)$ be a unary semi-automaton
 with $\Sigma = \{a\}$ and $S \subseteq Q$.
 Then $|\delta(S, a^k)| = 1$ for some $k \ge 0$
 if and only if $|\delta(S, a^{|Q|-1})| = 1$.
\end{restatable}

\begin{restatable}[]{proposition}{vectoratmostoneinftyatmostoneone}\label{prop:vector_at_most_one_infty_at_most_one_one}
 Let $\Sigma = \{a_1, \ldots, a_k\}$ be our alphabet.
 Consider the Constrained Synchronization Problem \ref{def:problem_L-constr_Sync}.
 Suppose the commutative constraint language $L$ is
 decomposed as stated in Theorem \ref{thm:reg_commutative_form},
 \begin{equation}\label{eqn:union_L_single_infty_single_one}
  L = \bigcup_{i=1}^n U_1^{(i)} \shuffle \ldots \shuffle U_k^{(i)}.
 \end{equation}
 Denote by $N = \{ (n_1^{(i)}, \ldots, n_k^{(i)} \mid i = 1, \ldots, n \}$
 the vector representation, according 
 to Definition \ref{def:constraint_vector} and corresponding
 to Equation \eqref{eqn:union_L_single_infty_single_one}.
 If for all $i \in \{1,\ldots, n\}$ 
 in the vector $(n_1^{(i)}, \ldots, n_k^{(i)})$,
 at most one entry equals $\infty$
 and at most one entry is non-zero, and if
 so it equals one, 
 then the problem is solvable in polynomial time. 
\end{restatable}

\subsection{The $\NP$-complete Variants of the Problem}
\label{sec:NP}

In this section, we state a criterion,
in terms of the constraint language, which gives $\NP$-hardness.
Surprisingly, in contrast to Proposition \ref{prop:vector_at_most_one_infty_at_most_one_one},
if some letter, whose appearance is bounded in 
an infinite language of the form $U_1^{(i)} \shuffle \ldots \shuffle U_k^{(i)}$,
is allowed to appear more than once, then we get \NP-hardness. 

\begin{restatable}[]{proposition}{singleletterunboundedtwobounded} \label{prop:single_letter_unbounded_two_bounded}
 Let $\Sigma = \{a_1, \ldots, a_k\}$ be our alphabet.
 Consider the Constrained Synchronization Problem \ref{def:problem_L-constr_Sync}.
 Suppose the commutative constraint language $L$ is
 decomposed as stated in Theorem \ref{thm:reg_commutative_form},
 \begin{equation}\label{eqn:union_L_np_hard_case}
  L = \bigcup_{i=1}^n U_1^{(i)} \shuffle \ldots \shuffle U_k^{(i)}.
 \end{equation}
 Denote by $N$ the vector representation, according 
 to Definition \ref{def:constraint_vector} and corresponding
 to Equation \eqref{eqn:union_L_np_hard_case}. 
 Suppose we find $i_0 \in \{1,\ldots, k\}$
 and a \emph{maximal}\footnote{Note that, for example, a commutative regular language
 with vector representation $N = \{ (2,\infty,0), (\infty,\infty,0) \}$ would give
 a constrained problem in \PTIME.} 
 vector $(n_1^{(i_0)}, \ldots, n_k^{(i_0)}) \in N$
 such that at least one of the following conditions is true:
 \begin{enumerate}
     \item[(i)] $n_{j_0}^{(i_0)} = \infty$ and $2\le n_{j_1}^{(i _0)} < \infty$
            for distinct $j_0, j_1 \in \{1,\ldots, k\}$, or
     \item[(ii)] $n_{j_0}^{(i_0)} = \infty$ and $1\le n_{j_1}^{(i_0)}, n_{j_2}^{(i_0)} < \infty$
            for distinct $j_0, j_1, j_2 \in \{1,\ldots, k\}$.
 \end{enumerate}
 Then the problem is \NP-hard.
\end{restatable}

\subsection{The $\PSPACE$-complete Variants of the Problem}
\label{sec:PSPACE}

\begin{restatable}[]{proposition}{oneletterboundedtwounbounded} \label{prop:one_letter_bounded_two_unbounded}
 %
 %
 %
 %
 
 Let $\Sigma = \{a_1, \ldots, a_k\}$ be our alphabet.
 Consider the Constrained Synchronization Problem \ref{def:problem_L-constr_Sync}.
 Suppose the commutative constraint language $L$ is
 decomposed as stated in Theorem \ref{thm:reg_commutative_form},
 \begin{equation}\label{eqn:union_L_np_hard_case3}
  L = \bigcup_{i=1}^n U_1^{(i)} \shuffle \ldots \shuffle U_k^{(i)}.
 \end{equation}
 Denote by $N$ the vector representation, according 
 to Definition \ref{def:constraint_vector} and corresponding
 to Equation \eqref{eqn:union_L_np_hard_case3}.
 Suppose we find $i_0 \in \{1,\ldots, n\}$ and distinct $j_0, j_1, j_2 \in \{1,\ldots, k\}$
 and a \emph{maximal} vector $(n_1^{(i_0)}, \ldots, n_k^{(i_0)}) \in N$
 such that
 $n_{j_0}^{(i_0)} = n_{j_1}^{(i_0)} = \infty$ and $1\le n_{j_2}^{(i_0)} < \infty$.
 Then the problem is $\PSPACE$-hard.
\end{restatable}

\subsection{Main Theorem}
\label{sec:main_thm}

Combining everything up to now gives our main computational complexity classification result
for $L(\mathcal B)\textsc{-Constr-Sync}$.

\begin{restatable}[]{theorem}{completeclassification} \label{thm:complete_classification}
 Let $\Sigma = \{a_1, \ldots, a_k\}$ be our alphabet.
 Consider the Constrained Synchronization Problem \ref{def:problem_L-constr_Sync}.
 Suppose the commutative constraint language $L$ is
 decomposed as stated in Theorem \ref{thm:reg_commutative_form},
 \begin{equation}\label{eqn:union_L_complete_classification}
  L = \bigcup_{i=1}^n U_1^{(i)} \shuffle \ldots \shuffle U_k^{(i)}.
 \end{equation}
 Denote by $N = \{ (n_1^{(i)}, \ldots, n_k^{(i)}) \mid i = 1, \ldots, n \}$
 the vector representation, according 
 to Definition \ref{def:constraint_vector} and corresponding
 to Equation \eqref{eqn:union_L_complete_classification}. By taking
 the maximal vectors in $N$, which is no restriction by Proposition~\ref{prop:vector_rep_incomparable}, we can assume
 the vectors in $N$ are incomparable. 
 \begin{enumerate}
   
     \item[(i)] 

         Suppose for all $i \in \{1,\ldots, n\}$,
         if we have distinct $j_0, j_1 \in \{1,\ldots, k\}$
         with $n_{j_0}^{(i)} = n_{j_1}^{(i)} = \infty$, then $n_j^{(i)} \in \{0,\infty\}$
         for all other $j \in \{1,\ldots, k\} \setminus \{j_0, j_1\}$.
         More formally,
         \begin{multline*}
          \forall i \in \{1,\ldots, n \} : ( \exists j_0, j_1 \in \{1,\ldots, k\} : j_0 \ne j_1 \land n_{j_0}^{(i)} = n_{j_1}^{(i)} = \infty ) \\  \rightarrow ( \forall j \in \{1,\ldots, k\} : n_j^{(i)} \in \{0,\infty\} ). 
         \end{multline*}
         Furthermore, suppose $N$
         fulfills the condition mentioned in Proposition \ref{prop:single_letter_unbounded_two_bounded},
         then it is $\NP$-complete.
         
     \item[(ii)] If the set $N$  fulfills the condition imposed by Propostion \ref{prop:one_letter_bounded_two_unbounded},
      then it is $\PSPACE$-complete.
    \item[(iii)] In all other cases the problem is in $\PTIME$.
 \end{enumerate}
\end{restatable}

The assumption that the vectors in $N$ are incomparable
is essential in the statement, otherwise it would be more complex.
For example, a language
with the vector representation 
$N = \{ (2,\infty,\infty,0),(\infty,\infty,\infty,0),(0,1,\infty,1)\}$
gives an \NP-complete constrained problem. 
However, the formula stated in Theorem~\ref{thm:complete_classification} 
for the \NP-complete case is not fulfilled, as the first vector has two entries
with $\infty$ and another non-zero finite entry.
But for $\{ (\infty,\infty,\infty,0),(0,1,\infty,1) \}$, the maximal vectors,
the conditions in the \NP-complete case above apply.
We give some examples for all cases in Example \ref{ex:classification}.

{\footnotesize  
\begin{example} \label{ex:classification} 
 Let $\Sigma = \{a,b,c\}$ with $a = a_1, b = a_2, c = a_3$.
\begin{itemize} 
\item If $L = \{aa\}\shuffle b(bb)^*$ with $N = \{(2,\infty,0)\}$,
 then $L\textsc{-Constr-Sync}$ is $\NP$-complete.
\item If $L = \{a\}\shuffle b(bb)^*\shuffle \{c\}$
 with $N = \{(1,\infty,1)\}$, then $L\textsc{-Constr-Sync}$ is $\NP$-complete.
\item The constraint language from Example \ref{ex:vector} gives
 a $\NP$-complete problem.
\item If $L = \{aa\}\shuffle b(bb)^* \cup (aaa)^* \shuffle b \shuffle c^*$ with $N = \{(2,\infty,0), (\infty, 1, \infty)\}$, then $L\textsc{-Constr-Sync}$ is $\PSPACE$-complete.
\item If $L = \{a\}\shuffle b(bb)^*$ with $N = \{(1,\infty,0)\}$,
 then $L\textsc{-Constr-Sync} \in \PTIME$.
\item If $L = (aa)^* \shuffle c(ccc)^*$ with $N = (\infty,0,\infty)$, then $L\textsc{-Constr-Sync} \in \PTIME$.
\end{itemize}
\end{example}
}

\subsection{Deciding the Computational Complexity of the Constrained Synchronization Problem} 
\label{sec:decision_problem}


This section addresses the issue of deciding
the computational complexity of $L(\mathcal B)\textsc{-Constr-Sync}$, 
for a constraint automaton such that $L(\mathcal B)$
is commutative.
The next definition is a mild generalization
of a definition first given in \cite{GomezA08}, and used
for state complexity questions in \cite{DBLP:conf/cai/Hoffmann19,Hoffmann20}.

\begin{definition} \label{def:comm_aut}
 Let $\Sigma = \{a_1, \ldots, a_k\}$
 and suppose $\mathcal A = (\Sigma, Q, \delta, s_0, F)$
 is a complete and deterministic automaton accepting a commutative language.
 Set $Q_j = \{ \delta(s_0, a_j^i) : i \ge 0 \}$
 for $j \in \{1,\ldots, k\}$. 
 The automaton $\mathcal C_{\mathcal A} = (\Sigma, Q_1 \times \ldots \times Q_k, \mu, t_0, E)$
 with $t_0 = (s_0, \ldots, s_0)$,
 $$
  \mu( w, (s_1, \ldots, s_k) )
   = (\delta(s_1, a_1^{|w|_{a_1}}), \ldots, \delta(s_k, a_k^{|w|_{a_k}}))
 $$
 and $E = \{ (\delta(t_0, a_1^{|w|_{a_1}}), \ldots , \delta(t_0, a_k^{|w|_{a_k}})) : w \in L(\mathcal A) \}$
 is called the \emph{commutative automaton} constructed from $\mathcal A$.
\end{definition}

If $\mathcal A$ is the minimal automaton of a commutative 
language, it is exactly the definition 
from \cite{GomezA08,DBLP:conf/cai/Hoffmann19,Hoffmann20}. In that case, also in \cite{GomezA08,DBLP:conf/cai/Hoffmann19,Hoffmann20}, it was shown that $L(\mathcal C_{\mathcal A}) = L(\mathcal A)$, and 
that $L(\mathcal A)$ is a union of certain shuffled languages. Both statements still hold
for any automaton $\mathcal A$ such that $L(\mathcal A)$
is commutative.


\begin{restatable}[]{theorem}{commautlanguage}\label{thm:comm_aut_language}
 Let $\Sigma = \{a_1, \ldots, a_k\}$ 
 and suppose $\mathcal A = (\Sigma, Q, \delta, s_0, F)$
 is a complete and deterministic automaton accepting a commutative language.
 Denote by $\mathcal C_{\mathcal A} = (\Sigma, Q_1 \times \ldots \times Q_k, \mu, t_0, E)$ the commutative automaton
 from Definition \ref{def:comm_aut}.
 Then $L(\mathcal C_{\mathcal A}) = L(\mathcal A)$.
\end{restatable}

The set of words that lead into a single state of the commutative automaton has
a simple form.

\begin{restatable}[]{lemma}{wordsinsinglestate}\label{lem:words_in_single_state} 
 Let $\Sigma = \{a_1, \ldots, a_k\}$
 and suppose $\mathcal A = (\Sigma, Q, \delta, s_0, F)$
 is a complete and deterministic automaton accepting a commutative language.
 Denote by $\mathcal C_{\mathcal A} = (\Sigma, Q_1 \times \ldots \times Q_k, \mu, t_0, E)$ 
 the commutative automaton from Definition \ref{def:comm_aut}.
 Let $s = (s_1, \ldots, s_k) \in Q_1 \times \ldots \times Q_k$
 and set $U_j = \{ u \in \{a_j\}^* \mid \delta(s_0, u) = s_j \}$.
 Then
 $$
  \{ w \in \Sigma^* \mid \mu(t_0, w) = (s_1, \ldots, s_k) \} 
   = U_1 \shuffle \ldots \shuffle U_k.
 $$
\end{restatable} 

{\footnotesize  
\begin{example}
 Note that the form from Lemma \ref{lem:words_in_single_state} need not hold for some arbitrary automaton.
 For example, let $\Sigma = \{a,b\}$ and $L = \Sigma^+$.
 Then a minimal automaton has two states with a single accepting state, 
 and the commutative automaton derived from it has four states, with three accepting states.
 We have $L = a^+ \cup b^+ \cup a^+ \shuffle b^+$.
\end{example}}

As the language of any deterministic automaton could be written as a disjoint union
of languages which lead into a single final state, the next is implied.

\begin{corollary}\label{cor:shuffle_union_form}
 Let $\Sigma = \{a_1, \ldots, a_k\}$
 and suppose $\mathcal A = (\Sigma, Q, \delta, s_0, F)$
 is a complete and deterministic automaton accepting a commutative language.
 Denote by $\mathcal C_{\mathcal A} = (\Sigma, Q_1 \times \ldots \times Q_k, \mu, t_0, E)$ 
 the commutative automaton from Definition \ref{def:comm_aut}.
 Suppose $E = \{ (s_1^{(l)}, \ldots, s_k^{(l)}) \mid l \in \{1,\ldots, m\} \}$
 for some $m \ge 0$.
 Set\footnote{If we start with the minimal automaton, then these are the same sets $U_j^{(l)}$ as introduced in \cite{Hoffmann20}.}
 $U_j^{(l)} = \{ u \in \{a_j\}^* \mid \delta(s_0, u) = s_j^{(l)} \}$
 for $l \in \{1,\ldots, m\}$ and $j \in \{1,\ldots, k\}$.
 Then
 \begin{equation}\label{eqn:shuffle_union_form}
  L(\mathcal A) = \bigcup_{l=1}^m  U_1^{(l)} \shuffle \ldots \shuffle U_k^{(l)}.
 \end{equation}
\end{corollary}

With these notions, we can derive a decision procedure. First construct the
commutative automaton. Then derive a representation as given in Equation \eqref{eqn:shuffle_union_form}.
Use this representation to compute a vector representation according to Definition \ref{def:constraint_vector}.
With the help of Theorem \ref{thm:complete_classification}, from such a vector representation
the computational complexity could be read off.

\begin{restatable}[]{theorem}{decisionprocedure}\label{thm:decision_procedure} 
    Let $\Sigma = \{a_1, \ldots, a_k\}$ be a fixed alphabet.
    For a given 
    (partial) automaton $\mathcal B = (\Sigma, P, \mu, p_0, F)$
    accepting a commutative language,
    the computational complexity of $L(\mathcal B)\textsc{-Constr-Sync}$
    could be decided in polynomial time.
\end{restatable}


\vspace*{-0.4cm}

\section{Conclusion}

We have looked at the Constrained Synchronization Problem~\ref{def:problem_L-constr_Sync} for commutative regular constraint languages, thereby
continuing the investigation started in \cite{DBLP:conf/mfcs/FernauGHHVW19}.
The complexity landscape for regular commutative constraint languages is completely understood.
Only the complexity classes $\PTIME$, $\NP$ and $\PSPACE$ arise, and
we have given conditions for $\PTIME$, $\NP$-complete and $\PSPACE$-complete problems.
In \cite{DBLP:conf/mfcs/FernauGHHVW19} the questions was raised if we can find constraint languages
that give other levels of the polynomial time hierarchy. At least for commutative regular languages
this is not the case. Lastly, we have given a procedure to decide the computational complexity
of $L(\mathcal B)\textsc{-Constr-Sync}$, for a given automaton $\mathcal B$ accepting a commutative language.

\section*{Acknowledgement}
{
\footnotesize  
I thank Prof. Dr.
Mikhail V. Volkov for suggesting the problem of constrained synchronization
during the workshop `Modern Complexity Aspects of Formal Languages' that took place  at Trier University  11.--15.\ February, 2019.
The financial support of this workshop by the DFG-funded project FE560/9-1 is gratefully acknowledged.
I thank my supervisor, Prof. Dr. Henning Fernau, for accepting me as I am, do not judge me on, and always listening to, my sometimes dumb musings; and for giving valuable feedback, discussions and research suggestions concerning the content of this article. }

\bibliographystyle{splncs04}
\bibliography{ms}

\clearpage
\section{Appendix}
Here we collect some proofs not given in the main text.
For establishing  some of our results,
we need the following computational problem taken from \cite{BerlinkovFS18}, which is a
\PSPACE-complete problem for at least binary alphabets, also see \cite{Rys83,San2005}.

\begin{decproblem}\label{def:sync-into-subset}
  \problemtitle{\textsc{Sync-Into-Subset}}
  \probleminput{Det. complete semi-automaton $\mathcal A = (\Sigma, Q, \delta)$ and $S \subseteq Q$.}
  \problemquestion{Is there a word $w \in \Sigma^*$ with $\delta(Q, w) \subseteq S$?}
\end{decproblem}

\begin{remark} \ref{def:sync-into-subset}
	The terminology is not homogeneous in the literature. For instance, \textsc{Sync-Into-Subset} 
	has a different name in \cite{BerlinkovFS18} and in \cite{Rys83}.
\end{remark}

We will also need the next problem from \cite{Koz77}, which is $\PSPACE$-complete in general, but $\NP$-complete
for unary automata, see \cite{fernau2017problems}.

\begin{decproblem}\label{def:problem_AutInt}
  \problemtitle{\textsc{Intersection-Non-Emptiness}}
  \probleminput{Deterministic complete automata $\mathcal A_1$, $\mathcal A_2$, \ldots, $\mathcal A_k$.}
  \problemquestion{Is there a word accepted by them all?}
\end{decproblem}

For some semi-automaton (or DFA) with state set $Q$ and transition function $\delta : Q \times \Sigma \to Q$,
a state $q$ is called a \emph{sink state}, if for all $x \in \Sigma$ we have $\delta(q,x) = q$.

\subsection{Proof of Proposition~\ref{lem:gen:inNP} (See page~\pageref{lem:gen:inNP})}
\lemgeninNP*
\begin{proof} 
 
 As finite languages are regular, and the shuffle operation preserves regular
 languages, the language $L$ is regular.
 Let $\mathcal B = (\Sigma, P, \mu, p_0, F)$
 be some partial automaton with $L(\mathcal B) = L$.
 First note that no final state could be a sink state,
 as then other letters than $a$ could appear infinitely often.
 Further, we can assume for each state $p \in P$
 we have some $u \in \Sigma^*$ with $\mu(p, u) \in F$.
 For otherwise we could drop this state and all transitions to it and get
 another partial automaton that still accepts the same language.
 Also we can assume that each state is reachable, i.e.,
 for $p \in P$ we have $u \in \Sigma$ with $\mu(p_0, u) = p$.
 Now suppose for $p \in P$ that $L(\mathcal B_{p,\{p\}})$
 is infinite. Choose $u,v \in \Sigma^*$
 with $\mu(p_0, u) = p$ and $\mu(p, v) \in F$.
 Then if $w \in L(\mathcal B_{p, \{p\}})$
 we have $uw^*v \subseteq L(\mathcal B)$.
 This gives $w \subseteq \{a\}^*$,
 as otherwise, if $|w|_b > 0$ for some $b \in\Sigma\setminus\{a\}$, then
 for each $n > 0$ we would have $|uw^nv|_b > n$.
 But by Definition of $L$ every letter distinct from $a$
 could only appear a bounded number of times. $\qed$
\end{proof}

\subsection{Proof of Proposition~\ref{prop:infinite_then_N} (See page~\pageref{prop:infinite_then_N})}
\infinitethenN*
\begin{proof} 
 Notation as in the statement of the proposition.
 Because $L \subseteq L'$ one direction is clear.
 Conversely
 suppose we have some synchronizing word $w \in L'$
 and assume $w \in V_1^{(i)} \shuffle \ldots \shuffle V_k^{(i)}$.
 If $i \ne i_0$, then as $V_j^{(i)} = U_j^{(i)}$
 for $j \in \{1,\ldots, k\}$ we have $w \in L$.
 So suppose $i = i_0$.
 As $U_{j_0}^{(i_0)}$
 is infinite, we have some $m \ge 0$
 such that $a_{j_0}^{|w|_{a_{j_0} + m}} \in U_{j_0}^{(i_0)}$.
 This gives
 $$ 
  w a_{j_0}^m \in U_{1}^{(i_0)} \shuffle \ldots \shuffle U_k^{(i_0)}
 $$
 as $a_j^{|w|_{a_j}} \in U_j^{(i_0)} = V_{j}^{(i_0)}$
 for $j \in \{1,\ldots, k\} \setminus \{j_0\}$.
 Hence $w \in L$ and by Theorem \ref{thm:add-stuff} the claim follows. $\qed$
\end{proof}

\subsection{Proof of Proposition~\ref{prop:vector_rep_incomparable} (See page~\pageref{prop:vector_rep_incomparable})}
\vectorrepincomparable*
\begin{proof} 
 Notation as in the statement of the proposition.
 Suppose we have some synchronizing word $w \in L$.
 If $w \in U_1^{(i)} \shuffle \ldots \shuffle U_k^{(i)}$
 with $i \ne i_0$, then also $w \in L'$.
 So suppose $w \in U_1^{(i_0)} \shuffle \ldots \shuffle U_k^{(i_0)}$.
 Let $y = (y_1^{(i_1)}, \ldots, y_k^{(i_1)})$
 with $i_1 \in \{1,\ldots, n\}\setminus \{i_0\}$
 and corresponding part $U_1^{(i_1)} \shuffle \ldots \shuffle U_k^{(i_1)} \subseteq L'$.
 As $x \le y$ for each $a_j^{|w|_{a_j}} \in U_j^{(i_0)}$
 with $j \in \{1,\ldots, k\}$
 we find $m_j \ge 0$ such that $a_j^{|w|_{a_j}} a_j^{m_j} \in U_j^{(i_1)}$.
 Hence $w a_1^{m_1} \cdots a_1^{m_k} \in U_1^{(i_1)} \shuffle \ldots \shuffle U_k^{(i_1)} \subseteq L'$  
 and by Theorem \ref{thm:add-stuff} the claim follows.~$\qed$
\end{proof}

\subsection{Proof of Lemma~\ref{lem:union} (See page~\pageref{lem:union})}
\lemunion*
\begin{proof} 
 Notation as in the statement. 
 The proof for $\mathcal X = \PTIME$
 works by checking in polynomial time all the languages $L_i$ in order, 
 which is a polynomial time operation\footnote{Actually, setting up a machine
 that runs a fixed number of other machines is a constant time operations in itself,
 as soon as one machine ends, enter the starting configuration of the next and so on.
 Hence essentially only the running time of the individual machines determines
 the total running time or space requirements. And here the language $L$
 and hence the value $n$ is part of the fixed constraint language.}.
 The same argument  gives the claim for $\mathcal X = \NP$. 
 This does not use nondeterminism, alterantively
 we could use nondeterminism by guessing $1 \le i \le n$
 first, and then checking for synchronizability in $L_i$.
 For $\mathcal X = \PSPACE$
 the same procedure of checking the languages $L_i$
 in order will work, as running a machine for each $L_i$
 one after another only needs a constant amount of extra instructions, 
 and as each machine only needs polynomial space the the total procedure
 will only use polynomial space. Alternatively
 we can use $\NPSPACE = \PSPACE$ by Savitch's Theorem \cite{Savitch70}
 and guess the language $L_i$. $\qed$
\end{proof}

\begin{lemma}\label{lem:cycle_states}
 Let $\mathcal A = (\Sigma, Q, \delta)$ be a unary semi-automaton
 with $\Sigma = \{a\}$. Then the set $T \subseteq Q$ of states appearing
 on some cycle of $\mathcal A$ is characterized
 by being maximal with the condition $\delta(T, a) = T$.
\end{lemma}
\begin{proof} 
 First some general remarks. The automaton graph of a unary automaton
 is the functional graph of the function given by the 
 single letter seen as a transformation on the states. Such graphs
 are sometimes called directed maximal pseudoforests, and they consists of cycles, and directed paths that must all end in some cycle. In \cite{DBLP:conf/mfcs/FernauGHHVW19}, these
 where also called sun-structures. 
 Note that for each state $q \in Q$, the state $\delta(q, a^{|Q|-1})$
 must always lie on some cycle of the mentioned graph, by the  pigeonhole principle.
 Also, if $q\in Q$ is a state from some cycle, then the state $\delta(q, a^i)$
 for some $i \ge 0$ is also contained on the same cycle. Hence, if $T$ denotes
 the set of all states on the cycles, we have $\delta(T, a) \subseteq T$.
 But, also conversely, if $q \in Q$ is a state on some cycle, it is the preimage
 of the direct predecessor in the cycle, hence $T \subseteq \delta(T, a)$.
 But, the condition $\delta(T, a) = T$ implies
 that only cycle states are contained in $T$.
 To see this, let $T_0 = \{ q \in Q \mid q\notin \delta(Q, a) \}$.
 %
 As cycle states are mapped to cycle states, if $q \in Q$
 is not on some cycle, and $q \notin T_0$, all
 states $S$ with $\{q\} = \delta(S, a)$ are also not on any cycle.
 If $\delta(T, a) = T$, then $a$ acts surjective on this set,
 hence $T \cap T_0 = \emptyset$. Suppose $q \in T$ is some state not on any
 cycle, and chose $S_i \subseteq Q$ maximal with $\{q\} = \delta(S_i, a^i)$.
 By assumption $T \cap S_i \ne \emptyset$ for each $i \ge 0$,
 as $T\cap S_1 \ne \emptyset$, and then arguing inductively for all $i \ge 0$.
 By finiteness, we must have $q \in S_j \cap S_i$ with $j > i$,
 but this implies $\delta(q,a^{i+j}) = \delta(q, a^i)$, i.e.,
 the state $\delta(q, a^i)$ is contained in some cycle. A contradiction,
 hence $T$ could not contain any states not on some cycle.
 Lastly, adding a cycle state $q$ to $T$, and with it the whole cycle $\{ \delta(q, a^i) \mid i \ge 0 \}$,
 the resulting set still obeys the equation. Hence if it is already maximal, it must contain
 every cycle state.
\end{proof}

\subsection{Proof of Lemma~\ref{lem:length_sync_from_word_unary} (See page~\pageref{lem:length_sync_from_word_unary})}
\lengthsyncfromwordunary*
\begin{proof} 
 Suppose $|\delta(S, a^k)| = 1$ for some $k \ge 0$.
 Choose $T \subseteq Q$ maximal such that $\delta(T, a) = T$, by Lemma~\ref{lem:cycle_states} precisely those states
 on the cycles of $\mathcal A$.
 Set $R := \delta(S, a^{|Q|-1})$.
 We have $R \subseteq T$,
 as for any $q \in S$  by finiteness
 $$\delta(q, a^{|Q|}) \in \{ q, \delta(q, a), \ldots, \delta(q, a^{|Q|-1}) \},$$
 which implies we reached some cycle. As the letter $a$ acts as a permutation
 on the set $T$ we have $|R| = |\delta(R, a^i)|$
 for each $i \ge 0$. Hence we must have $|R| = 1$. $\qed$
\end{proof}

\subsection{Proof of Proposition~\ref{prop:vector_zero_infty} (See page~\pageref{prop:vector_zero_infty})}
\vectorzeroinfty*
\begin{proof} 
 By Proposition \ref{prop:infinite_then_N},
 we can assume that if $U_{j}^{(i)}$ is infinite 
 with $i \in \{1,\ldots,n\}$ and $j \in \{1,\ldots, k\}$
 we have $U_{j}^{(i)} = \{a_j\}^*$. By assumption, every letter
 in $U_1^{(i)} \shuffle \ldots \shuffle U_k^{(i)}$ either appears not at all, or infinitely often, which
 by the above means without any restriction.
 Hence, 
 $$
   U_1^{(i)} \shuffle \ldots \shuffle U_k^{(i)} = \Gamma^*
 $$
 for some\footnote{Note
 that $U_1^{(i)} \shuffle \ldots \shuffle U_k^{(i)} = \{\varepsilon\}$
 is possible, which corresponds to the vector $(0,\ldots,0)$.
 But this is covered by $\Gamma = \emptyset$, as by definition $\emptyset^* = \{\varepsilon\}$.} 
 $\Gamma \subseteq \Sigma$. The constrained synchronization problem
 for each single language $\Gamma^*$ could be solved in polynomial time.
 Just ignore all transitions by letters in $\Sigma \setminus \Gamma$
 of any input semi-automaton. The resulting unconstrained
 synchronization problem could then be solved in polynomial time
 by Theorem \ref{thm:unrestricted_sync_poly_time}.
 By Lemma \ref{lem:union}
 the original problem could be solved in polynomial time. $\qed$
\end{proof}

\subsection{Proof of Proposition~\ref{prop:vector_at_most_one_infty_at_most_one_one} (See page~\pageref{prop:vector_at_most_one_infty_at_most_one_one})}
\vectoratmostoneinftyatmostoneone* 
\begin{proof}
 By Proposition \ref{prop:infinite_then_N},
 we can assume that if $U_{j}^{(i)}$ is infinite 
 with $i \in \{1,\ldots,n\}$ and $j \in \{1,\ldots, k\}$
 we have $U_{j}^{(i)} = \{a_j\}^*$.
 By Lemma  \ref{lem:union}
 we can consider a single language of the 
 form $U_1^{(i)} \shuffle \ldots \shuffle U_k^{(i)}$.
 If in the corresponding vector no $\infty$ appears, this language
 is finite. This case is solvable in polynomial time by Lemma \ref{lem:finite}.
 If only a single entry equals $\infty$, and all others are zero,
 then this is solvable in polynomial time by Proposition \ref{prop:vector_zero_infty}.
 So assume we have $j_0, j_1 \in \{1,\ldots, k\}$
 with $U_{j_0}^{(i)} = \{a_{j_0}\}^*$,
 $U_{j_1}^{(i)} = \{a_{j_1}\}$
 and $U_j^{(i)} = \{\varepsilon\}$
 for $j \in \{1,\ldots, k\}\setminus\{j_0, j_1\}$.
 
 Let $\mathcal A = (\Sigma, Q, \delta)$ be a semi-automaton.
 By the constrained language, only the letters $a_{j_0}$
 and $a_{j_1}$ could appear in any synchronizing word.
 For abbreviation we write $a$ for $a_{j_0}$ and $b$ for $a_{j_1}$.
 We can assume $\Sigma = \{a,b\}$
 by ignoring all other transitions.
 The letter $b$ must appear precisely once.
 First, let us only consider the transitions labelled with $a$, i.e., view $\mathcal A$
 as a unary automaton over $\{a\}$.
 
 Choose $T \subseteq Q$ maximal such that $\delta(T, a) = T$, which, by Lemma~\ref{lem:cycle_states}, are precisely those states
 on the $a$-cycles\footnote{These are the cycles that we end up when we 
 start in any state and just read in the letter $a$, i.e., those
 cycles that are exclusively labelled by words from $a^*$.} of $\mathcal A$.
 As $T \subseteq Q$ we have $T \subseteq \delta(Q, a^i)$ for each $i \ge 0$.
 Also, with the same argument as in the proof of Lemma \ref{lem:length_sync_from_word_unary},
 we have $\delta(Q, a^{|Q|-1}) \subseteq T$.
 Taken together $T = \delta(Q, a^{|Q|-1})$, which gives $T = \delta(Q, a^i)$
 for each $i \ge |Q| - 1$.
 So to see if we have any word of the form $a^i b a^j$ 
 with $|\delta(Q, a^i b a^j)| = 1$,
 we just have to test all words with $0 \le i \le |Q|-1$,
 and, by applying Lemma \ref{lem:length_sync_from_word_unary}
 to $S = \delta(Q, a^ib)$, we only have to test $j = |Q| - 1$.
 In total we only need to test $|Q|$ words $ba^{|Q|-1}, aba^{|Q|-1}, \ldots, a^{|Q|-1}ba^{|Q|-1}$
 and each could be done in polynomial time. $\qed$
\end{proof}

\subsection{Proof of Proposition~\ref{prop:single_letter_unbounded_two_bounded} (See page~\pageref{prop:single_letter_unbounded_two_bounded})}
\singleletterunboundedtwobounded*
\begin{proof}
  Notation as in the statement of the Proposition.
 The proofs for both cases (i) and (ii) are very similar.
 We will give a full proof for case (i) and then describe where it
 has to be altered to give a proof for case (ii).
 
  (i) By Proposition~\ref{prop:vector_rep_incomparable},
   we can take the maximal vectors in $N = \{ (n_1^{(i)}, \ldots, n_k^{(i)}
 \mid i \in \{1,\ldots, n\} \}$, which does 
   not change the computational complexity.
   Hence, by taking the maximal vectors, we can assume that the
   vectors in $N$ are incomparable.
 Note that if we take the maximal vectors in $N$, the assumptions
 of the statement do not change. Hence it is unaffected
 by this assumption respectively modification of $N$.
 We write 
 $$
  L = \bigcup_{i=1}^n U_1^{(i)} \shuffle \ldots \shuffle U_k^{(i)}
 $$
 in correspondence with the set $N$ according to Definition~\ref{def:constraint_vector}.
 
 First a rough outline of the reduction that we will construct.
 Please see Figure~\ref{fig:reduction_np} for a drawing of our reduction
 in accordance with the notation that will be introduced in this proof.
 
 \newcommand{\automatonpic}[3]{%
	\begin{scope}
		\node [rectangle,draw,thick,text width=7cm,minimum height=2cm,
	       text centered,rounded corners, fill=white, name = re] {};
	   \node[state] (s) at (-3,0) {$#2$};
	   \node[state] (t) at (-5,0) {$#1$};
	   \node at (-1,0) {\LARGE $#3$};
	      \draw (2,0) ellipse (1cm and 0.75cm);
	      \node at (1.5,0) {$F$};
	 \path[->] (t) edge [above] node {$b$} (s);
	    
	    \node[state] (q1) at (-2,0.5) {};
	    \node[state] (q2) at (0.5,-0.5) {};
	    
	    \path[->] (q1) edge [bend right=30,above] node {$b$} (s);
	    \path[->] (q2) edge [bend left=30, above] node {$b$} (s);
\end{scope}}    

\begin{figure}[htb]
     \centering
    \scalebox{.7}{
\begin{tikzpicture}
 \tikzset{every state/.style={minimum size=1pt},>=stealth'}
 \tikzstyle{dotted}=                  [dash pattern=on \pgflinewidth off 2pt]
   \node (cloud) at (0,0)  {\tikz \automatonpic{t_1}{s_1}{\mathcal A_1};};
   \node (cloud) at (0,2.5) {\tikz \automatonpic{t_2}{s_2}{\mathcal A_2};};
   \node (cloud) at (0,7.5) {\tikz \automatonpic{t_l}{s_l}{\mathcal A_l};};
 
   \draw [line width=2pt, line cap=round, dash pattern=on 0pt off 2\pgflinewidth] (0,4) -- (0,6);
   
 \node[state] (s1) at (6,8) {$r_1$};
 \node[state] (s2) at (8,8) {$r_2$};
 \node[state] (s3) at (10,8) {$r_3$};
 \node[state] (s4) at (12.5,8) {$r_m$};
 \node (s5) at (11,8) {};
 \node (s6) at (11.5,8) {};
 
 \node[state] (p21) at (10.2, 2) {$p_{2,1}$};
 \node[state] (p22) at (10.2, 4) {$p_{2,2}$};
 \node[state] (p2m) at (10.2, 7) {$p_{2,m(2)}$};
 \node (p21a) at (10.2, 5.3) {};
 \node (p21b) at (10.2, 5.7) {};
 \path[->] (p21) edge [left] node {$a_{\lambda(2)}$} (p22);
 \path[->] (p22) edge [left] node {$a_{\lambda(2)}$} (p21a);
 \draw[dashed] (10.2,5.4) -- (10.2,5.8);
 \path[->] (p21b) edge [left] node {$a_{\lambda(2)}$} (p2m);
 \path[->] (p2m)  edge [below] node {$a_{\lambda(2)}$} (s4);

 \node[state] (p31) at (12,2) {$p_{3,1}$};
 \node[state] (p32) at (12,4) {$p_{3,2}$};
 \node[state] (p3m) at (12,6.5) {$p_{3,m(3)}$};
 \node (p31a) at (12, 5.3) {};
 \node (p31b) at (12, 5.7) {};
 \path[->] (p31) edge [left] node {$a_{\lambda(3)}$} (p32);
 \path[->] (p32) edge [left] node {$a_{\lambda(3)}$} (p31a);
 \draw[dashed] (12,5.4) -- (12,5.8);
 \path[->] (p31b) edge [left] node {$a_{\lambda(3)}$} (p3m);
 \path[->] (p3m)  edge [right] node {$a_{\lambda(3)}$} (s4);   
 
 \path[->] (s1) edge [above] node {$b$} (s2); 
 \path[->] (s2) edge [above] node {$b$} (s3);
 \path[->] (s3) edge [above] node {$b$} (s5);
 \path[->] (s6) edge [above] node {$b$} (s4);
 
 \draw[dashed] (11,8) -- (11.5,8);
 
 \path[->] (s4) edge [loop above] node {$a,b,c$} (s4);

 \node[state] (a1) at (3,7.75) {};
 \node[state] (a2) at (3,2.75) {};
 \node[state] (a3) at (3,0.25) {};
 
 \path[->] (a1) edge [bend right=10,above,pos=0.8] node {$b$} (s1);
 \path[->] (a2) edge [bend right=30,above,pos=0.4] node {$b$} (s1);
 \path[->] (a3) edge [bend right=30,above,pos=0.4] node {$b$} (s1);
 
\end{tikzpicture}}

   \caption{Schematic illustration of the reduction in the proof of Proposition \ref{prop:single_letter_unbounded_two_bounded}
  for $\Sigma = \{a,b,c\} = \{a_1, a_2, a_3\}$ and a language of the form 
  $
      L = U_1^{(1)} \shuffle U_2^{(1)} \shuffle U_3^{(1)} \cup 
        U_1^{(2)} \shuffle U_2^{(2)} \shuffle U_3^{(2)} \cup  
        U_1^{(3)} \shuffle U_2^{(3)} \shuffle U_3^{(3)}.
  $
  Here $i_0 = 1$ with $U_1^{(1)} = \{a\}^*$, $U_2^{(1)} = \{b^m\}^*$ 
   and $U_3^{(1)} = \{ c \}$ for $2 \le m < \infty$.
   }
  \label{fig:reduction_np}
\end{figure}
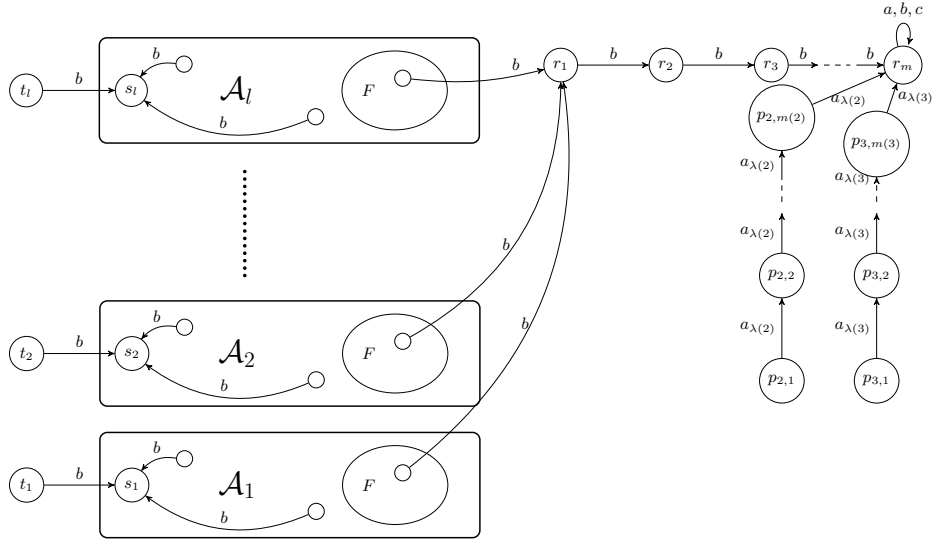

 We will use the problem~\textsc{Intersection-Non-Emptiness}
 from Definition~\ref{def:problem_AutInt}, which is $\NP$-complete
 for unary alphabets.
 We construct a set $P$ of states that guarantees
 we use the set $U_1^{(i_0)} \shuffle \ldots \shuffle U_k^{(i_0)}$
 for permissible synchronizing words. We do this because the property of it having
 one letter that could occur arbitrary often, and one letter to appear a 
 specific, strictly greater than one, number of times, is crucial.
 The letter that is unrestricted is the letter over which the input automata
 are defined, the restricted letter is used to enforce that we have 
 a word that is accepted by them all.

 By incomparability of the vectors in $N$ for 
 each $i \in \{1,\ldots, n\}\setminus\{i_0\}$
 there exists some index $j \in \{1,\ldots, k\}$ such that
 \begin{equation}\label{eqn:incomparable_index_np} 
  n_{j}^{(i_0)} > n_{j}^{(i)}.
 \end{equation}
 We define a function $\lambda : \{1,\ldots, n\}\setminus\{i_0\} \to \{1,\ldots, k\}$
 by choosing such an index $j$, i.e., setting $\lambda(i) = j$ for each $i \in  \{1,\ldots, n\}\setminus\{i_0\}$ with 
 the $j \in \{1,\ldots, k\}$ from Equation~\eqref{eqn:incomparable_index_np}.
 We will use these indices 
 to distinguish the corresponding sets
 used in Definition~\ref{def:constraint_vector}
 $$
  U_1^{(i)} \shuffle \ldots \shuffle U_k^{(i)}
 $$
 from the set
 $$
  U_1^{(i_0)} \shuffle \ldots \shuffle U_k^{(i_0)}.
 $$

 We define a function $m : \{1,\ldots, n\}\setminus \{i_0\} \to (\mathbb N_0 \setminus\{0\})$
  which will be used later
  to single out $U_1^{(i_0)} \shuffle \ldots \shuffle U_k^{(i_0)}$
  by setting\footnote{Note that $n_{\lambda(i)}^{(i_0)} = \infty$ implies $n_{\lambda(i)}^{(i)} \ne \infty$ by Equation~\eqref{eqn:incomparable_index_np},}
  for $i \in \{1,\ldots, n\}\setminus \{i_0\}$
  \begin{equation}\label{eqn:def_mi_np}
   m(i) = \left\{ \begin{array}{ll}
     n_{\lambda(i)}^{(i_0)} & \mbox{ if } n_{\lambda(i)}^{(i_0)} \ne \infty \\ 
     n_{\lambda(i)}^{(i)} + 1 & \mbox{ if } n_{\lambda(i)}^{(i_0)} = \infty.
   \end{array}\right.
  \end{equation}
  For words $w \in \Sigma^*$ 
  with $|w|_{a_{\lambda(i)}} \ge m(i)$ for all $i \in \{1,\ldots, i\}\setminus\{i_0\}$,
  the following holds
  \begin{equation}\label{eqn:single_componet_np_proof}
   w \in L \Leftrightarrow w \in U_1^{(i_0)} \shuffle \ldots \shuffle U_k^{(i_0)}.
  \end{equation}
  As for each $i \in \{1,\ldots, n\}\setminus \{i_0\}$
  we have some $\lambda(i)$ such that $U_{\lambda(i)}^{(i)}$ is finite
  and contains a unique word of length $n_{\lambda(i)}^{(i)}$, and
  $$
   w \notin U_{\lambda(i)}^{(i)}
  $$
  by choice of $|w|_{a_{\lambda(i)}}$.
  Set 
  $$ 
   P = \{ p_{i, 1}, \ldots, p_{i, m(i)} \mid i \in \{1,\ldots, n\} \setminus \{i_0\}\}.
  $$

 By Proposition~\ref{prop:infinite_then_N}
 we can assume that if $U_{j}^{(i)}$ is infinite 
 with $i \in \{1,\ldots,n\}$ and $j \in \{1,\ldots, k\}$
 we have $U_{j}^{(i)} = \{a_j\}^*$.
 In what follows only  the letters
 $a_{j_0}$ and $a_{j_1}$ are essential.
 We denote by $a$ the letter $a_{j_0}$,
 by $b$ the letter $a_{j_1}$.
 We also set $m = n_{j_1}^{(i_0)}$ for abbreviation.
 
 We use the problem~\textsc{Intersection-Non-Emptiness} from Definition \ref{def:problem_AutInt}
 for unary automata, which is NP-complete in this case, for our reduction.
  Let $\mathcal A_1,\ldots, \mathcal A_l$
  be automata with $\mathcal A_i = (\{a\}, Q_i, \delta_i, s_i, F_i)$ 
  for $i \in \{1,\ldots, l\}$ and disjoint state sets.
  Construct a semi-automaton $\mathcal A = (\Sigma, Q, \delta)$ with state set
  $Q = Q_1 \cup \ldots \cup Q_l \cup \{t_1, \ldots, t_l, r_1, \ldots, r_{m-1}\} \cup P$ 
  and transition function 
  $$
   \delta(q, a) = \left\{ 
    \begin{array}{ll}
    \delta_i(q, a)  & \mbox{ if } q \in Q_i, \\
     q              & \mbox{ if } q \in \{t_1, \ldots, t_k, r_1, \ldots, r_{m-1}\}, 
   \end{array}\right.
  $$
  and
  $$
   \delta(q, b) = \left\{ \begin{array}{ll}
    s_i & \mbox{ if } q \in \{t_i\} \cup  S_i \setminus F_i, \\
    r_1 & \mbox{ if } q \in F_i, \\ 
    r_{i+1} & \mbox{ if } q = r_i \mbox{ for } i \in \{1,\ldots, m-2\}, \\ 
    q       & \mbox{ if } q = r_{m-1}. \\ 
    \end{array}\right.
  $$
  For $i \in \{1,\ldots n\} \setminus \{i_0\}$
  and $r \in \{1,\ldots,m(i)\}$ set
  $$
   \delta( p_{i, r}, a_{\lambda(i)} ) 
    = \left\{
    \begin{array}{ll}
      p_{i, r + 1} & \mbox{ if } r < n_{\lambda(i)}^{(i_0)} \\
      r_{m-1}      & \mbox{ if } r = n_{\lambda(i)}^{(i_0)}
    \end{array}
    \right.
  $$
  and $\delta( p_{i, r}, a_j ) = p_{i,r}$ for $j \ne \{1,\ldots, k\}\setminus\{\lambda(i)\}$.
  Lastly for $q \in Q \setminus P$ we
  set $\delta(q, c) = q$ for each $c \in \Sigma \setminus \{a,b\}$. Then our
  automaton is fully specified.

  We argue that our semi-automaton $\mathcal A$ 
  has a synchronizing word 
  in $L$
  if and only if $\bigcap_{i=1}^k L(\mathcal A_i) \ne \emptyset$.
  
  First suppose $a^n \in \bigcap_{i=1}^k L(\mathcal A_i)$.
  Then it is easy to see that $\delta(Q, ba^nb^{m-1}) = \{r_{m-1}\}$.
  We have $b^{m} \in U_{j_1}^{(i_0)}$
  and $a^n \in U_{j_0}^{(i_0)}$.
  For $j \in \{1,\ldots, k\} \setminus \{j_0, j_1\}$
  choose any $u_j \in U_j^{(i_0)}$. Let $u$
  be the concatenation of all these words in any order.
  Then we have $\delta(Q, ba^nb^{m-1}u) = \{r_{m-1}\}$
  and $b a^n b^{m-1} u \in U_1^{(i_0)} \shuffle \ldots \shuffle U_k^{(i_0)} \subseteq L$.

  Conversely assume we have $w \in L$ 
  with $|\delta(Q, w)| = 1$.
  As $r_{m-1}$ is a sink state we have $\delta(Q, w) = \{ r_{m-1} \}$.
  We need one $b$  to leave any state from $\{t_1, \ldots, t_l\}$.
  After this we end up in some state from $\{s_1, \ldots, s_l\} \subseteq Q_1\cup \ldots \cup Q_l$.
  And from those states to get to $r_1$, then $r_2$ and so on until $r_{m - 1}$
  we have to read $m - 1$ additional times the letter $b$.
  Hence, a word that could map any state in $\{t_1, \ldots, t_l\}$
  to $r_{m - 1}$ has to contain at least $m$ many times the letter $b$.
  
  For some $i \in \{1,\ldots, n\}$ we have $w \in U_1^{(i)} \shuffle \ldots \shuffle U_k^{(i)}$.
  Consider the states in $P$. The only
  way to go from $p_{i,1}$ to $r_{m-1}$ for $i \in \{1,\ldots, n\}\setminus\{i_0\}$
  is to read at least $m(i)$ times the letter $a_{\lambda(i)}$.
  Hence $|w|_{a_{\lambda(i)}} \ge m(i)$ and so by Equation \eqref{eqn:single_componet_np_proof}
  we have $i = i_0$. But as $U_{j_1}^{(i)}$ contains a unique word
  of length $m$ and with $b = a_{j_1}$
  we have $|w|_b = m$.
  
  Write $w = u_0 b u_1 b \cdots b u_{m} v$ with $u_i \in (\Sigma \setminus\{b\})^*$ for $i \in \{0,\ldots, m\}$.

  By construction $\{t_1, \ldots, t_k\} \subseteq \delta(Q, u_0)$.
  Hence by definition of the transition function
  $$
   \{s_1, \ldots, s_k, r_2, \ldots, r_{m-1}\} \subseteq \delta(Q, u_0b) \setminus P \subseteq \{s_1, \ldots, s_k, r_1, \ldots, r_{m-1}\}.
  $$
  Note that for any $q \in Q \setminus \{r_1, \ldots, r_{m-1}\}$
  and $u \in (a^* b a^*)^r$ with $r < m-1$ we have
  \begin{equation}\label{eqn:b_count}
   \delta(q, u) \subseteq Q \setminus \{r_{r+1}, \ldots, r_{m-1} \}.
  \end{equation}
  Assume $q = \delta(s_i, u_1) \notin F_i$ for some $i \in \{1,\ldots, l\}$,
  then by Equation \eqref{eqn:b_count} as $\delta(q, b) = s_i$,
  we have 
  $$
  \delta(s_i, u_1 b u_2 b \cdots b u_m) \subseteq Q \setminus \{r_{m-1}\}.
  $$
  Hence $\delta(s_i, u_1) \in F_i$ for $i \in \{1,\ldots, k\}$.
  As by construction of $\mathcal A$ only the letter $a$ and $b$ 
  act non-trivial\footnote{Meaning as non-identity transformations
  on the state set under consideration.} on the state set $Q \setminus P$, $u_1$ does not contain the letter $b$ and no state from $P$ could be entered from any state in $Q \setminus P$,
  in particular not from $s_i$, which implies $\delta(s_i, u) \in Q\setminus P$
  for each prefix of $u$ of $u_1$,
  we have that $\delta(s_i, a^{|u_1|}) = \delta(s_i, u_1)$.
  This gives $a^{|u_1|} \in \bigcap_{i=1}^k L(\mathcal A_i)$.

  (ii) In this case let $a = a_{j_0}, b = a_{j_1}$ and $c = a_{j_2}$.
   Set $m = n_{j_2}^{(i_0)}$.
   We can use essentially the same reduction. The difference is that 
   we use the letter $b$ to reset all automata $\mathcal A_1, \mathcal A_2, \ldots, \mathcal A_l$
   to their initial states.
   Instead of $m-1$ states $r_1, \ldots, r_{m-1}$ we use $m$ states $r_1, \ldots, r_m$,
   and the letter $c$ is used to move from state $r_i$ to state $r_{i+1}$ until we reach
   the final sink state $r_m$. All other letters induce self-loops on the states $r_1, \ldots, r_m$.
   Also inside the automata $\mathcal A_1, \ldots, A_l$ the letter $b$ also moves
   every state to the corresponding start state. The letter $c$ is used to move from any final
   state to the state $r_1$. For non-final states the letter $c$ induces a self-loop.
   With this construction, we could argue similar to case (i) that the thus altered automaton construction
   admits a synchronizing word in the constraint language
   if and only if we have a unary word accepted by all input automata.  $\qed$
\end{proof}

\subsection{Proof of Proposition~\ref{prop:one_letter_bounded_two_unbounded} (See page~\pageref{prop:one_letter_bounded_two_unbounded})}
\oneletterboundedtwounbounded*
\begin{proof} 
 Notation as in the statement of the proposition.
 By Proposition~\ref{prop:vector_rep_incomparable},
   we can take the maximal vectors in $N = \{ (n_1^{(i)}, \ldots, n_k^{(i)}
 \mid i \in \{1,\ldots, n\} \}$, which does 
   not change the computational complexity.
   Hence, by taking the maximal vectors, we can assume that the
   vectors in $N$ are incomparable.
 Note that if we take the maximal vectors in $N$, the assumptions
 of the statement do not change. Hence it is unaffected
 by this assumption respectively modification of $N$.
 We write 
 $$
  L = \bigcup_{i=1}^n U_1^{(i)} \shuffle \ldots \shuffle U_k^{(i)}
 $$
 in correspondence with the set $N$ according to Definition \ref{def:constraint_vector}.

 First, a rough outline of the reduction that we will construct.
 Please see Figure~\ref{fig:reduction_pspace} for a drawing of our reduction
 in accordance with the notation that will be introduced in this proof.
 
 \newcommand{\automatacloudother}[2][.44]{%
	\begin{scope}[#2]
		\node [rectangle,draw,thick,text width=10cm,minimum height=7cm,
		text centered,rounded corners, fill=white, name = re] {};
\end{scope}}    

\newcommand{\setSbar}[2][.35]{%
\begin{scope}[#2]
\pgftransformscale{#1}%
\pgfpathmoveto{\pgfpoint{261 pt}{115 pt}} 
  \pgfpathcurveto{\pgfqpoint{70 pt}{107 pt}}
                 {\pgfqpoint{127 pt}{291 pt}}
                 {\pgfqpoint{180 pt}{260 pt}} 
  \pgfpathcurveto{\pgfqpoint{78 pt}{382 pt}}
                 {\pgfqpoint{381 pt}{445 pt}}
                 {\pgfqpoint{412 pt}{440 pt}}
  \pgfpathcurveto{\pgfqpoint{577 pt}{587 pt}}
                 {\pgfqpoint{698 pt}{488 pt}}
                 {\pgfqpoint{685 pt}{366 pt}}
  \pgfpathcurveto{\pgfqpoint{840 pt}{192 pt}}
                 {\pgfqpoint{610 pt}{157 pt}}
                 {\pgfqpoint{610 pt}{157 pt}}
  \pgfpathcurveto{\pgfqpoint{531 pt}{39 pt}}
                 {\pgfqpoint{298 pt}{51 pt}}
                 {\pgfqpoint{261 pt}{115 pt}}
\pgfusepath{fill,stroke}       
\end{scope}}  

\begin{figure}[htb]
     \centering
    \scalebox{.65}{
\begin{tikzpicture}
 \tikzset{every state/.style={minimum size=1pt},>=stealth'}
 \node (cloud) at (0,0) {\tikz \automatacloudother{fill=gray!0,thick};};
 \node (cloud) at (1,-0.3) {\tikz \setSbar{fill=gray!20,opacity=0.5, thick};};
 
 \node[state] (s1) at (6,3) {$s_1$};
 \node[state] (s2) at (8,3) {$s_2$};
 \node[state] (s3) at (10,3) {$s_3$};
 \node[state] (s4) at (12.5,3) {$S_m$};
 \node (s5) at (11,3) {};
 \node (s6) at (11.5,3) {};
 
 
 \node[state] (p21) at (10.2,-3) {$p_{2,1}$};
 \node[state] (p22) at (10.2,-1) {$p_{2,2}$};
 \node[state] (p2m) at (10.2, 2) {$p_{2,m(2)}$};
 \node (p21a) at (10.2, 0.3) {};
 \node (p21b) at (10.2, 0.7) {};
 \path[->] (p21) edge [left] node {$a_{\lambda(2)}$} (p22);
 \path[->] (p22) edge [left] node {$a_{\lambda(2)}$} (p21a);
 \draw[dashed] (10.2,0.4) -- (10.2,0.8);
 \path[->] (p21b) edge [left] node {$a_{\lambda(2)}$} (p2m);
 \path[->] (p2m)  edge [below] node {$a_{\lambda(2)}$} (s4);

 \node[state] (p31) at (12,-3) {$p_{3,1}$};
 \node[state] (p32) at (12,-1) {$p_{3,2}$};
 \node[state] (p3m) at (12, 1.5) {$p_{3,m(3)}$};
 \node (p31a) at (12, 0.3) {};
 \node (p31b) at (12, 0.7) {};
 \path[->] (p31) edge [left] node {$a_{\lambda(3)}$} (p32);
 \path[->] (p32) edge [left] node {$a_{\lambda(3)}$} (p31a);
 \draw[dashed] (12,0.4) -- (12,0.8);
 \path[->] (p31b) edge [left] node {$a_{\lambda(3)}$} (p3m);
 \path[->] (p3m)  edge [right] node {$a_{\lambda(3)}$} (s4);   
 
 \path[->] (s1) edge [above] node {$c$} (s2); 
 \path[->] (s2) edge [above] node {$c$} (s3);
 \path[->] (s3) edge [above] node {$c$} (s5);
 \path[->] (s6) edge [above] node {$c$} (s4);
 
 \draw[dashed] (11,3) -- (11.5,3);
 
 \node at (-3,1.5) {\LARGE $\mathcal A$};
 \node at (3,1) {\LARGE $S$};
 
 \node[state] (a1) at (-3.5,1) {};
 \node[state] (a2) at (-2,2.4) {};
 \node[state] (a3) at (-4,-1.8) {};
 
 \node[state] (a4) at (0,0) {};
 \node[state] (a5) at (2,1) {};
 \node[state] (a6) at (0.5,-2) {};
 
 \path[->] (a4) edge [bend left=30,above] node {$c$} (s1);
 \path[->] (a5) edge [bend right=15,above,pos=0.6] node {$c$} (s1);
 \path[->] (a6) edge [bend right=30,right,pos=0.8] node {$c$} (s1);
  
 \path[->] (a1) edge [loop below] node {$c$} (a1); 
 \path[->] (a2) edge [loop below] node {$c$} (a2); 
 \path[->] (a3) edge [loop below] node {$c$} (a3); 
  
 \path[->] (s4) edge [loop above] node {$a,b,c$} (s4); 
  
\end{tikzpicture}}
  \caption{Schematic illustration of the reduction in the proof of Proposition \ref{prop:single_letter_unbounded_two_bounded}
  for $\Sigma = \{a,b,c\} = \{a_1, a_2, a_3\}$ and a language of the form 
  $
    L = U_1^{(1)} \shuffle U_2^{(1)} \shuffle U_3^{(1)} \cup 
         U_1^{(2)} \shuffle U_2^{(2)} \shuffle U_3^{(2)} \cup  
         U_1^{(3)} \shuffle U_2^{(3)} \shuffle U_3^{(3)}.
  $
  Here $i_0 = 1$ with $U_1^{(1)} = \{a\}^*$, $U_2^{(1)} = \{b\}^*$ 
  and $U_3^{(1)} = \{ c^m \}$ for $1 \le m < \infty$.
   }
  \label{fig:reduction_pspace}
\end{figure}
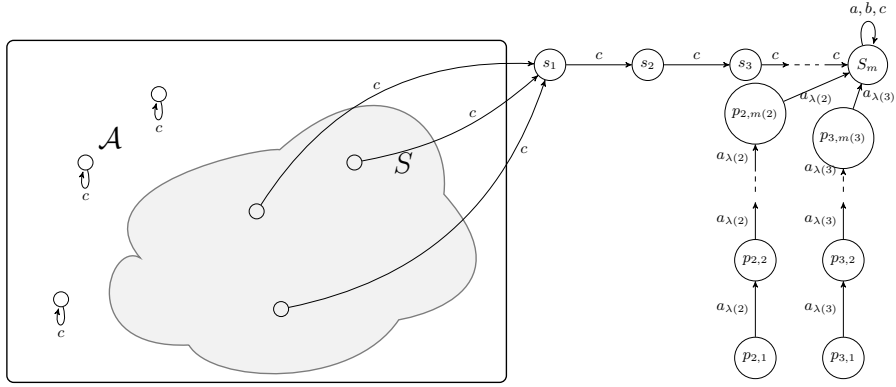

 We will use the problem \textsc{Sync-Into-Subset}
 from Definition \ref{def:sync-into-subset}, which is PSPACE-complete 
 for some fixed binary alphabet.
 We construct a set $P$ of states that guarantees
 we use the set $U_1^{(i_0)} \shuffle \ldots \shuffle U_k^{(i_0)}$
 for permissible synchronizing words. We do this because the property of it having
 two letters that could occur arbitrary often, and one letter is required to appear a specific
 non-zero number of times, is crucial.
 The two letters that are unrestricted are the letters over which some input automaton
 is defined, the restricted letter is used to enforce that we have some word over
 the unrestricted letters that maps all states into some specific set of states.

 By incomparability of the vectors in $N$ for 
 each $i \in \{1,\ldots, n\}\setminus\{i_0\}$
 there exists some index $j \in \{1,\ldots, k\}$ such that
 \begin{equation}\label{eqn:incomparable_index} 
  n_{j}^{(i_0)} > n_{j}^{(i)}.
 \end{equation}
 We define a function $\lambda : \{1,\ldots, n\}\setminus\{i_0\} \to \{1,\ldots, k\}$
 by choosing such an index $j$, i.e., setting $\lambda(i) = j$ for each $i \in  \{1,\ldots, n\}\setminus\{i_0\}$ with 
 the $j \in \{1,\ldots, k\}$ from Equation \eqref{eqn:incomparable_index}.
 We will use these indices 
 to distinguish the corresponding sets
 used in Definition \ref{def:constraint_vector}
 $$
  U_1^{(i)} \shuffle \ldots \shuffle U_k^{(i)}
 $$
 from the set
 $$
  U_1^{(i_0)} \shuffle \ldots \shuffle U_k^{(i_0)}.
 $$

 We also define a function $m : \{1,\ldots, n\}\setminus \{i_0\} \to (\mathbb N_0 \setminus\{0\})$
  which will be used later
  to single out $U_1^{(i_0)} \shuffle \ldots \shuffle U_k^{(i_0)}$
  by setting\footnote{Note that $n_{\lambda(i)}^{(i_0)} = \infty$ implies $n_{\lambda(i)}^{(i)} \ne \infty$
  by Equation~\eqref{eqn:incomparable_index}.}
  for $i \in \{1,\ldots, n\}\setminus \{i_0\}$
  \begin{equation}\label{eqn:def_mi}
   m(i) = \left\{ \begin{array}{ll}
     n_{\lambda(i)}^{(i_0)} & \mbox{ if } n_{\lambda(i)}^{(i_0)} \ne \infty \\ 
     n_{\lambda(i)}^{(i)} + 1 & \mbox{ if } n_{\lambda(i)}^{(i_0)} = \infty.
   \end{array}\right.
  \end{equation}
  For a word $w \in \Sigma^*$ 
  with $|w|_{a_{\lambda(i)}} \ge m(i)$
  for all $i \in \{1,\ldots, i\}\setminus\{i_0\}$,
  the following holds
  \begin{equation}\label{eqn:single_componet_pspace_proof}
   w \in L \Leftrightarrow w \in U_1^{(i_0)} \shuffle \ldots \shuffle U_k^{(i_0)}.
  \end{equation}
  As for each $i \in \{1,\ldots, n\}\setminus \{i_0\}$
  we have some $\lambda(i)$ such that $U_{\lambda(i)}^{(i)}$ is finite
  and contains a unique word of length $n_{\lambda(i)}^{(i)}$, and
  $$
   w \notin U_{\lambda(i)}^{(i)}
  $$
  by choice of $|w|_{a_{\lambda(i)}}$.
  Set 
  $$ 
   P = \{ p_{i, 1}, \ldots, p_{i, m(i)} \mid i \in \{1,\ldots, n\} \setminus \{i_0\}\}.
  $$

 By Proposition \ref{prop:infinite_then_N}
 we can assume that if $U_{j}^{(i)}$ is infinite 
 with $i \in \{1,\ldots,n\}$ and $j \in \{1,\ldots, k\}$
 we have $U_{j}^{(i)} = \{a_j\}^*$.
 In what follows only  the letters
 $a_{j_0}, a_{j_1}$ and $a_{j_2}$ are essential.
 We denote by $a$ the letter $a_{j_0}$,
 by $b$ the letter $a_{j_1}$
 and by $c$ the letter $a_{j_2}$.
 We also set $m = n_{j_2}^{(i_0)}$ for abbreviation.

 Now our reduction from \textsc{Sync-Into-Subset}
 given in Definition \ref{def:sync-into-subset}.
 Set $\Gamma = \{a,b\}$.
 Let $\mathcal A = (\Gamma, Q, \delta)$ be a semi-automaton with non-empty subset $S \subseteq Q$.
 We construct an automaton $\mathcal A' = (\Sigma, Q', \delta')$
 with  $Q' = Q \cup P \cup \{ s_1, \ldots, s_m \}$.
 
 For states $q \in Q' \setminus P$ we set
 $$
  \delta'(q, x) = \left\{
  \begin{array}{ll}
   \delta(q, x) & \mbox{ if } q \in Q \mbox{ and } x \in \Gamma, \\ 
   q            & \mbox{ if } x \in \Sigma \setminus \{a,b,c\},  \\
   q            & \mbox{ if } q \in Q \setminus S \mbox{ and } x = c, \\
   s_1          & \mbox{ if } q \in S \mbox{ and } x = c, \\
   s_{i+1}      & \mbox{ if } q = s_i \mbox{ with } i \in \{1,\ldots,m-1\} \mbox{ and } x = c, \\
   s_m          & \mbox{ if } q = s_m \mbox{ and } x = c, \\
   q            & \mbox{ if } q \in \{s_1,\ldots,s_m\} \mbox{ and } x \in \{a,b\}.
  \end{array}
  \right.
 $$
 and for the states in $P$  with $i \in \{1,\ldots n\} \setminus \{i_0\}$
 and $r \in \{1,\ldots,m(i)\}$ we set
  $$
   \delta( p_{i, r}, a_{\lambda(i)} ) 
    = \left\{
    \begin{array}{ll}
      p_{i, r + 1} & \mbox{ if } r < m(i) \\
      s_{m}    & \mbox{ if } r = m(i) 
    \end{array}
    \right.
  $$
  and $\delta( p_{i, r}, a_j ) = p_{i,r}$ for $j \ne \lambda(i)$.

  We have that $\mathcal A'$ has a synchronizing word $w \in L$
  if and only if $\delta(Q, u) \subseteq S$ for some $u \in \Gamma^*$.
  
  First assume $\delta(Q, u) \subseteq S$  for some $u \in \Gamma^*$.
  Then $\delta'(Q \cup \{s_1, \ldots, s_m\}, uc^m) = \{ s_m \}$. 
  We define $u_j \in \Sigma \setminus \{c\}$ for $j \in \{1,\ldots k\}\setminus \{j_2\}$
  by setting\footnote{Note that $a_j \ne c$ if $j \ne j_2$.}
  $$
   u_j = \left\{ 
   \begin{array}{ll}
    \mbox{any} \in U_j^{(i_0)} & \mbox{ if } j \notin \lambda( \{1, \ldots, n\} \setminus \{i_0\} ), \\
    a_j^{m(i)} \in U_j^{(i_0)} & \mbox{ if } j = \lambda(i) \mbox{ for some } i \in \{1,\ldots, n\}\setminus \{i_0\}
   \end{array} \right.
  $$
  which is well-defined as $\lambda(i) = \lambda(i')$ for $i, i' \in  \{1, \ldots, n\} \setminus \{i_0\} )$
  implies $m(i) = m(i')$ by Equation \eqref{eqn:def_mi}.
  Let $v$ be the concatenation of the $u_j$ in any order
  and set
  $
   w = uc^m v.
  $  
  Then $w \in U_1^{(i_0)} \shuffle \ldots \shuffle U_k^{(i_0)} \subseteq L$. 
  Note that the factors $u_{j_0} \subseteq\{ a \}^*, u_{j_1} \subseteq\{b\}^*$ and $u \in \Gamma^*$
  of $v$ pose no problem here as $U_{j_0}^{(i_0)} = \{a\}^*$
  and $U_{j_1}^{(i_1)} = \{b\}^*$.
  Then by choice of the $u_j$ we have $\delta(P, w) = \{s_m\}$, and as $s_m$ is a sink state $\delta(Q', w) = \{s_m\}$.

  %
  Conversely, assume $\mathcal A'$ has a synchronizing word $w \in L$.
  As $s_m$ is a sink state\footnote{This is 
  a state $q \in Q'$ with $\delta'(q, x) = q$
  for all $x \in \Sigma$.}
  we must have $\delta'(Q', w) = \{s_m\}$.
  Also because $\delta(P, w) = \{s_m\}$ we 
  have $|w|_{a_{\lambda(i)}} \ge m(i)$ for each $i \in \{1,\ldots,n\}\setminus\{i_0\}$.
  So by Equation \eqref{eqn:single_componet_pspace_proof}
  this implies $w \in U_1^{(i_0)} \shuffle \ldots \shuffle U_k^{(i_0)} \subseteq L$.
  Hence as $U_{j_2}^{(i_0)}$ contains a unique word of length $m \ge 1$
  we have $|w|_{c} = m$. Write
  $w = u_0 c u_1 c\cdots c u_m$ with $u_i \in (\Sigma\setminus \{c\})^*$ 
  for $i \in \{1,\ldots, m\}$.
  For any $u \in \Sigma^*$ with $|u|_c < m$ we have
  $$ 
   \delta'(Q, u) \subseteq Q \setminus ( P \cup \{ s_{|u|_c+1}, \ldots, s_m \} ).
  $$
  But we reach $s_m$ so we must have $\delta'(Q, u_0c) \cap Q = \emptyset$,
  for otherwise we would not have enough letters $c$ left to transfer
  any state from $\delta'(Q, u_0c) \cap Q$ to $s_m$.
  The condition $\delta'(Q, u_0c) \cap Q = \emptyset$
  with $|u_0|_c = 0$
  is only possible if $\delta'(Q, u_0) \subseteq S$.
  As for $x \in \Sigma \setminus ( \Gamma \cup \{c\} )$
  we have $\delta(q, x) = q$ for each $q \in Q$, we can remove all these letters
  from $u_0$ to get a new word $u \in \Gamma^*$
  with $\delta(Q, u) = \delta'(Q, u) = \delta'(Q, u_0) \subseteq S$. $\qed$
\end{proof}

\subsection{Proof of Theorem~\ref{thm:complete_classification} (See page~\pageref{thm:complete_classification})}
\completeclassification*
\begin{proof}
 Notation as in the statement of the Theorem. By Proposition \ref{prop:infinite_then_N}
 we can assume that if $U_{j}^{(i)}$ is infinite,
 with $i \in \{1,\ldots,n\}$ and $j \in \{1,\ldots, k\}$,
 we have $U_{j}^{(i)} = \{a_j\}^*$.
 Both Proposition~\ref{prop:single_letter_unbounded_two_bounded}
 and Proposition~\ref{prop:one_letter_bounded_two_unbounded}
 give the corresponding hardness results for case (i) and (ii).
 By Theorem~\ref{thm:L-contr-sync-PSPACE}
 the problem is always in $\PSPACE$.
 This gives case~(ii).
 Suppose case~(i) holds. Beside hardness, we still have to show containment in $\NP$.
 We will show that for each 
 language $U_1^{(i)} \shuffle \ldots \shuffle U_k^{(i)}$ with $i \in \{1,\ldots, n\}$
 the constrained synchronization problem for this language is in $\NP$. 
 By Lemma \ref{lem:union} this would give our claim for case~(i).
 If two different languages $U_{j_0}^{(i)}$, $U_{j_1}^{(i)}$ with $j_0, j_1 \in \{1,\ldots, k\}$
 are infinite, then we can apply Proposition~\ref{prop:vector_zero_infty} by assumption
 from case~(i).
 Otherwise, either the language $U_1^{(i)} \shuffle \ldots \shuffle U_k^{(i)}$ with $i \in \{1,\ldots, k\}$
 is finite, in which case we can apply Lemma~\ref{lem:finite}, or a single language $U_{j_0}^{(i)}$ with $j \in \{1,\ldots, k\}$
 is infinite, in which case we can apply Lemma~\ref{lem:gen:inNP}, by the assumption that infinite
 languages $U_j^{(i)}$ with $i \in \{1,\ldots,n\}$ and $j \in \{1,\ldots,k\}$ equal $\{a_j\}^*$.
 Hence for all these languages the problem is in $\NP$.
 
 Now suppose case~(iii) holds. Then for each $(n_1^{(i)}, \ldots, n_k^{(i)}) \in N$
 with $i \in \{1,\ldots, n\}$
 one of the following conditions must hold, as otherwise
 we would be either in case (i) or (ii).
 \begin{enumerate}
 \item[(a)] If $n_{j_0}^{(i)} = n_{j_1}^{(i)} = \infty$ 
  for two distinct $j_0 \ne j_1$ with $j_0, j_1\in \{1,\ldots, k\}$
  then $n_j^{(i)} \in \{0,\infty\}$
  for all other $j \in \{1,\ldots, k\} \setminus\{ j_0, j_1 \}$.
 
 \item[(b)] If $n_{j_0}^{(i)} = \infty$
  and $n_{j}^{(i)} \ne \infty$ for $j \in \{1,\ldots, k\} \setminus \{j_0\}$.
  Then either $n_j^{(i)} = 0$ for all $j \in \{1,\ldots, k\} \setminus \{j_0\}$
  or $n_{j_1}^{(i)} = 1$ for some $j_1 \in \{1,\ldots, k\} \setminus \{j_0\}$
  and $n_j^{(i)} = 0$ for $j \in \{1,\ldots, k\}\setminus\{j_0, j_1\}$.
  
 \item[(c)] We have $n_j^{(i)} \ne \infty$ 
  for all $j \in \{1,\ldots, k\}$.
 \end{enumerate}
 We consider the  Constrained Synchronization
 Problem \ref{def:problem_L-constr_Sync}
 for the single language
 $$
  U_1^{(i)} \shuffle \ldots \shuffle U_k^{(i)}
 $$
 corresponding to the vector $(n_1^{(i)}, \ldots, n_k^{(i)})$
 and show that it is in $\PTIME$. In case (a) by Proposition \ref{prop:vector_zero_infty}
 the problem is in $\PTIME$.
 For case (b) by Proposition \ref{prop:vector_at_most_one_infty_at_most_one_one}
 the problem is in $\PTIME$. 
 In case (c) the corresponding
 language is finite, hence by Lemma \ref{lem:finite} in $\PTIME$.
 Taken together, by Lemma \ref{lem:union},
 the problem for $L$ is in $\PTIME$. $\qed$
\end{proof}

\subsection{Proof of Lemma~\ref{lem:words_in_single_state} (See page~\pageref{lem:words_in_single_state})}
\wordsinsinglestate*
\begin{proof}
 Notation as in the statement of the Lemma.
 First suppose $w \in \Sigma^*$
 with $\mu(t_0, w) = (s_1, \ldots, s_k)$.
 Then $\delta(s_0, a_j^{|w|_{a_j}}) = s_j$ for all $j \in \{1,\ldots, k\}$.
 Hence $a_j^{|w|_{a_j}} \in U_j$
 and as $w \in a_1^{|w|_{a_1}} \shuffle \ldots \shuffle a_k^{|w|_{a_k}}$
 we get $w \in U_1 \shuffle \ldots \shuffle U_k$.
 Conversely assume $w \in U_1 \shuffle \ldots \shuffle U_k$.
 Then as $|w|_{a_j} \in U_j$ we have
 $\delta(s_0, a_j^{|w|_{a_j}}) = s_j$ for all $j \in \{1,\ldots, k\}$.
 By definition this is equivalent with $\mu(t_0, w) = (s_1, \ldots, s_k)$. $\qed$
\end{proof}

\subsection{Proof of Theorem~\ref{thm:comm_aut_language} (See page~\pageref{thm:comm_aut_language})}
\commautlanguage*
\begin{proof}
 If $w \in L(\mathcal A)$ then by definition $\mu(t_0, w) \in E$, hence $w \in L(\mathcal C_{\mathcal A})$.
 Conversely suppose $w \in L(\mathcal C_{\mathcal A})$.
 Then $\mu(t_0, w) \in E$, which is equivalent with
 $\delta(s_0, a_j^{|w|_{a_j}}) = \delta(s_0, a_j^{|u|_{a_j}})$
 for some $u \in L(\mathcal A)$ and $j \in \{1,\ldots, k\}$.
 As $\delta(s_0, u) \in F$ and $L(\mathcal A)$ is commutative,
 we have\footnote{If $\mathcal A$ is the minimal automaton, then
 both states would be equal. Because it has the property that if $u$ is a permutation of $v$
 then $\delta(s_0, u) = \delta(s_0, v)$.
 For if $\delta(s_0, u) \ne \delta(s_0', v)$, then for one state, say
 $s = \delta(s_0, u)$, we would have some $w$ with $\delta(s, w) \in F$ and $\delta(\delta(s_0, v), w) \notin F$.
 But as $uw$ is a permutation of $vw$ this is not possible.
 But here $\mathcal A$ could be any automaton accepting the language, and the only thing that is retained
 under permuting letters is that, if we start in the start state, 
 either both words end in a final state or in a non-final state.} 
 $\delta(s_0, a_1^{|u|_{a_1}} a_2^{|u|_{a_2}} \cdots a_k^{|u|_{a_k}}) \in F$.
 This gives $$\delta(s_0, a_1^{|w|_{a_1}} a_2^{|u|_{a_2}} \cdots a_k^{|u|_{a_k}}) \in F$$
 as $\delta(s_0, a_1^{|u|_{a_1}}) = \delta(s_0, a_1^{|w|_{a_1}})$. Continuing similar
 $$
  \delta(s_0, a_2^{|u|_{a_2}} a_1^{|w|_{a_1}} \cdots a_k^{|u|_{a_k}}) \in F
 $$
 which gives $\delta(s_0, a_2^{|w|_{a_2}} a_1^{|w|_{a_1}} a_3^{|u|_{a_3}} \cdots a_k^{|u|_{a_k}}) \in F$.
 Doing this for all letters we find
 $$
  \delta(s_0, a_2^{|u|_{a_2}} a_1^{|w|_{a_1}} \cdots a_k^{|w|_{a_k}}) \in F
 $$
 which gives $\delta(s_0, w) \in F$, or $w \in L(\mathcal A)$. $\qed$
\end{proof}

\subsection{Proof of Theorem~\ref{thm:decision_procedure} (See page~\pageref{thm:decision_procedure})}
\decisionprocedure*
\begin{proof} 
 We can assume $\mathcal B$ is complete, otherwise we add a trap state. And if
 $\mu(s, a)$ is undefined for $s \in P$ and $a \in \Sigma$
 we add a transition to the trap state instead. This operation does not alters the accepted
 language.
 Construct the commutative automaton $\mathcal C_{\mathcal B}$
 which has at most $|Q|^{k}$ states. From it we can derive
 the form \eqref{eqn:shuffle_union_form} given in Corollary \ref{cor:shuffle_union_form}.
 From this form we can compute a vector set $N$
 according to Definition \ref{def:constraint_vector}, as it is easy
 to check if a unary language is finite or infinite.
 Also note that in this form the unary languages $U_j^{(l)}$
 could be accepted by unary automata with a single final state by the way 
 they are defined.
 Then $L(\mathcal B)$ is infinite if and only if in at least one vector
 the entry $\infty$ appears. The condition (i) 
 from Theorem \ref{thm:complete_classification}
 could be easily checked, also condition (ii).
 Hence by Theorem \ref{thm:complete_classification}
 this gives a decision procedure for the computational complexity
 of the resulting problem $L(\mathcal B)\textsc{-Constr-Sync}$.
 Every step could be performed in polynomial time. $\qed$
\end{proof}
\end{document}